\documentclass[opre,nonblindrev]{informs3} 

\OneAndAHalfSpacedXI
\usepackage{natbib}
\bibpunct[, ]{(}{)}{,}{a}{}{,}%
\def\bibfont{\small}%
\usepackage{booktabs} 

\usepackage[ruled]{algorithm2e} 

\usepackage{ bbm, color, enumerate, amsbsy, amsfonts, amsopn, amssymb,  amsxtra, bezier, graphicx, latexsym, verbatim,
	pictexwd, supertabular, url, dsfont,leftidx, setspace}
\usepackage{mathtools,bm}
\usepackage{multirow}
\usepackage{footnote}
\makesavenoteenv{tabular}
\makesavenoteenv{table}

\newcommand{\place}{{\sf Placement}}
\newcommand{\placel}{{\sf Placement on a Line}}

\newcommand{\onee}{\mathbbm{1}}
\newcommand{\ep}{\mathbb{E}}

\def\opt{\textsc{OPT}}
\def\alg{\textsc{ALG}}
\DeclarePairedDelimiter{\ceil}{\lceil}{\rceil}
\DeclarePairedDelimiter{\floor}{\lfloor}{\rfloor}

\SetAlFnt{\small}
\SetAlCapFnt{\small}
\SetAlCapNameFnt{\small}
\SetAlCapHSkip{0pt}
\IncMargin{-\parindent}

\newcounter{omarcounter}

\newcounter{rajancounter}

\TheoremsNumberedThrough     
\ECRepeatTheorems

\EquationsNumberedThrough

\begin{document}
	\TITLE{
   {When Location Shapes Choice: \\ Placement Optimization of Substitutable Products}
	}
	\ARTICLEAUTHORS{%
		\AUTHOR{Omar El Housni}
	 	\AFF{School of Operations Research and Information Engineering, Cornell Tech, Cornell University, \EMAIL{oe46@cornell.edu }} 
	 \AUTHOR{Rajan Udwani}
	 \AFF{Departnent of Industrial Engineering and Operations Research, UC Berkeley,  \EMAIL{rudwani@berkeley.edu}}
    \medskip

    \medskip

    \medskip

    \medskip
}
\ABSTRACT{%
	Strategic product placement can have a strong influence on customer purchase behavior in physical stores as well as online platforms. Motivated by this, we consider the problem of optimizing the placement of substitutable products in designated display locations to maximize the expected revenue of the seller. We model the customer behavior as a two-stage process: first, the customer visits a subset of display locations according to a browsing distribution; second, the customer chooses at most one product from the displayed products at those locations according to a choice model. Our goal is to design a general algorithm that can select and place the products optimally for \emph{any browsing distribution and choice model}, and we call this the \place\ problem. We give a randomized algorithm that utilizes an $\alpha$-approximate algorithm for cardinality constrained assortment optimization and outputs a $\frac{\Theta(\alpha)}{\log m}$-approximate solution (in expectation) for \place\ with $m$ display locations, i.e., our algorithm outputs a solution with value at least $\frac{\Omega(\alpha)}{\log m}$ factor of the optimal and this is tight in the worst case. We also give algorithms with stronger guarantees in some special cases. In particular, we give an efficient deterministic $\frac{\Omega(1)}{\log m}$-approximation algorithm for the Markov choice model, and a tight $(1-1/e)$-approximation algorithm for the problem when products have identical prices.
}

\KEYWORDS{Product Placement; Assortment Optimization; Choice Model; Approximation Algorithms}
\maketitle


\section{Introduction}





Assortment optimization is a fundamental problem in revenue management with applications in a variety of industries including retailing and online advertising. Given a universe of products, the goal in assortment optimization is to select a subset of products to maximize the seller's objective such as expected revenue or expected sales. The expected revenue (or sales) of a given assortment depends on the preferences of customers and this is typically captured using a discrete choice model. A variety of choice models have been proposed in the literature, and there is a large stream of work dedicated to the development of efficient algorithms for finding optimal or approximately optimal assortments  under different choice models. 

An underlying assumption in most discrete choice models is that 
a	customer's choice depends only on his/her preferences over the assortment of products. 
In practice, when a customer browses a store - whether physical or online - he/she typically encounters only a subset of the products on display. Customers 
choose based on their preferences over the products \emph{that they actually see}, and the likelihood of observing a given product depends on the location of this product in the store. {\color{black} High-visibility placements significantly boost a product's chances of being seen by customers in both physical stores \citep{larson2005exploratory} and online environments such as online advertising \citep{ ghose2009empirical, agarwal2011location} and e-commerce \citep{ursu2018power}. }
In general, customer choice behavior is influenced by the placement of products at different locations as well as the overall assortment of products on display. 


Product placement is particularly important in brick-and-mortar stores, where the arrangement of products on various shelves and prominent display areas such as mannequins and end caps (display locations at the end of aisles), can influence customer purchase behavior. By placing popular or high-margin products at eye level or in high-traffic areas, retailers can capture customers' attention, leading to impulse purchases and increased sales~(\cite{flamand2023store}). In fact, both physical and online retailers often display a given product at multiple locations in the store in order to increase the product's visibility (example shown in Figure \ref{fig:one}). More broadly, designing product placement is part of \emph{visual merchandising} in the retail industry, which is the practice of optimizing the presentation of products in a store to create an appealing and engaging environment for the customer. The practice is also common in the context of online browsing, where the optimal placement of product tiles or ads on a webpage affects the click behavior of the customer~(\cite{victor}). In several online settings, products are displayed in a vertical list and customers typically browse the list from top to bottom. The problem of optimizing a vertical list of products has been actively studied in the literature and it is referred to as \emph{display optimization}~(\cite{alidany}) or \emph{product framing}~(\cite{gallego2020approximation}). Many online platforms display a matrix of products instead of a vertical list. To the best our knowledge, no algorithmic results are known for more general settings where the display locations are not vertically stacked.	

\begin{figure}[h]
	\begin{center}
		\includegraphics[scale=0.85]{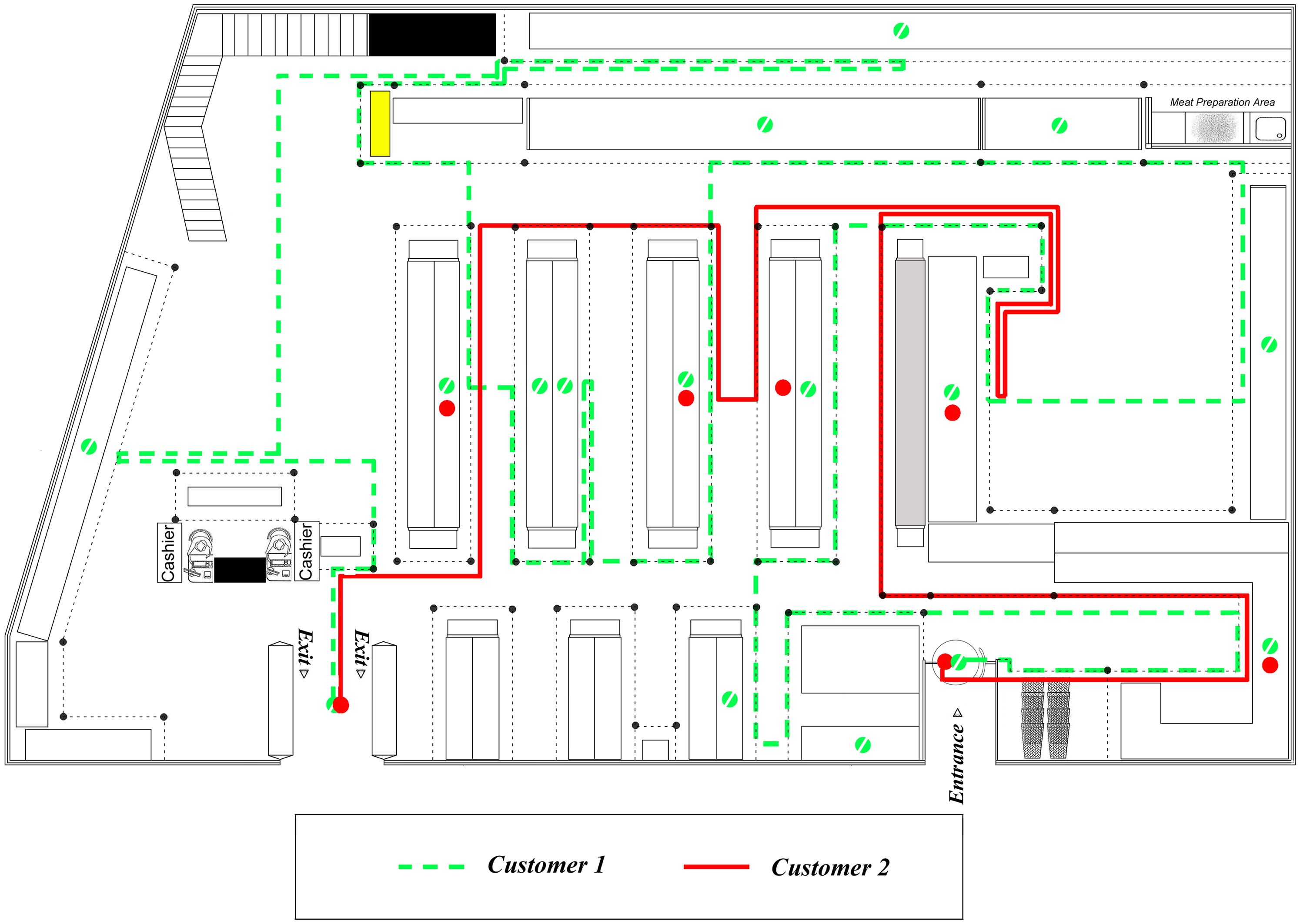}
	\end{center}
	\caption{Journeys of two different customers through a store. Figure reproduced from \cite{flamand2023store} with permission.}\label{fig:two}
\end{figure}

\begin{figure}
	\centering
	\includegraphics[scale=0.25]{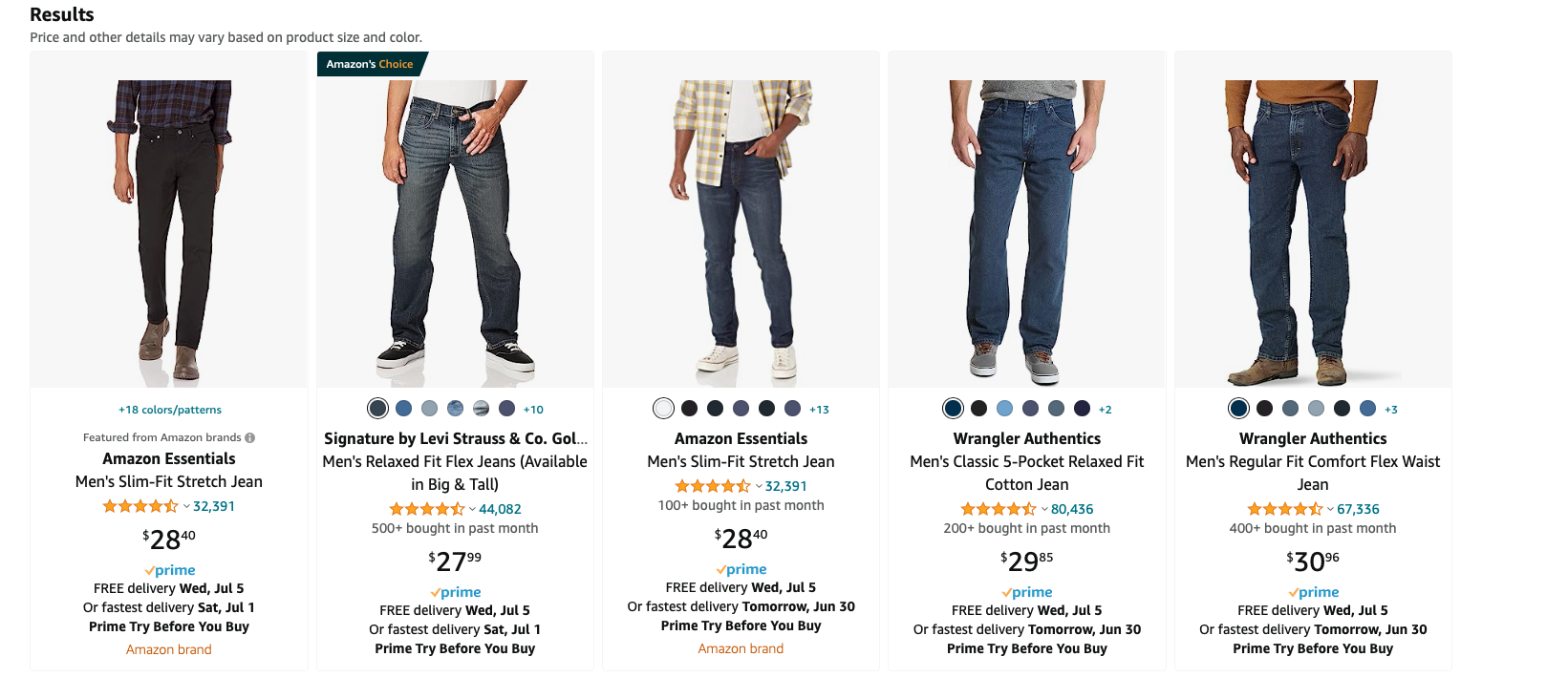}
	\caption{Screenshot from Amazon.com. The product on the extreme left (Amazon Essentials) is the same as the product in middle (also Amazon Essentials) and links to the same product page.}
	\label{fig:amazon}\label{fig:one}
\end{figure}

\vspace{2mm}
\noindent
{\bf Informal model description:} 
We propose a general formulation of the problem of placing substitutable products in a designated set of display locations in a store (physical or online). 

We model the customer behavior as a two-stage process: first, a customer entering the store visits a subset of display locations according to a \emph{browsing distribution}; second, the customer chooses from the products placed at the locations that they visited according to a discrete choice model. {\color{black} We refer to the probability distribution of the set of visited locations as \emph{the browsing distribution}}.  Note that customers may take different paths through the store (for example, see Figure \ref{fig:two}).  
	A customer's path may depend on the layout of the store and the customer's patience level and capacity for browsing but we assume that it does not depend on the product placement, i.e., the browsing distribution is independent of the placement of products.  We also assume that in the second stage, the customer's choice depends only on the set of products that they saw in the first stage and not on where these products were displayed, 
	i.e., the discrete choice model is also independent of the placement of products. Later, we generalize some of our results to settings where the choice model depends on the placement.
	
	The objective is to select an assortment of products and determine their placement in order to maximize the expected revenue from a customer visit. Each location holds at most one product but a given product may be placed at multiple display locations. { Our model also captures settings where multiple products may be placed at each location, as well as, the setting where locations cannot be left empty.} To capture the different types of choice behavior and browsing patterns that may arise in the diverse array of applications of this problem, we focus on designing a general algorithm that can select and place the products optimally for \emph{any browsing distribution and almost any choice model}. We call this the \place\ problem. 
	
	Our two-stage process is inspired by consider-then-choose choice models~(\cite{consider}), where the customer first shortlists a subset of the assortment (a consideration set) and then chooses from the subset according to a discrete choice model. The main difference is that our model endogenizes the consideration set of customers through their browsing behavior and the product placement in a store. 

\vspace{2mm}
\noindent
{\bf Fundamental challenges:} 
The \place\ problem involves two challenging and interrelated tasks: selecting the assortment of products to display in the store, and finding the optimal placement of this assortment over the available locations. 
Given an assortment, the optimal placement of products in the assortment depends on the manner in which customers enter and move around the store, i.e., the browsing distribution. { In fact, \place\ is NP-hard even in the special case where the locations are vertically stacked and customer chooses according to the Multinomial Logit (MNL) model (see Remark \ref{pline} and Lemma \ref{lem:hardness} in Section \ref{sec:prelim}).} In general, the browsing behavior can be quite diverse (for example, see Figure \ref{fig:two}) and this makes it significantly more challenging to find the optimal placement 
from the enormous number of	possible solutions: $s^m$ possible placements for an assortment of $s$ products and $m$ locations since every product may be repeated at multiple locations. 

In the setting where customers visit all locations with certainty (a trivial browsing distribution), the placement problem is equivalent to the classical assortment optimization problem with a cardinality constraint due to the limited number of display locations. For a general Random Utility Maximization (RUM)-based choice model \citep{block1959random}, this problem is computationally intractable and no non-trivial approximation guarantees are possible~(\cite{aouad2018approximability}). However, for many RUM-based { single-purchase choice models}, such as the Multinomial Logit (MNL) choice model, Nested Logit (NL) choice model, or the Markov choice model, there exist efficient algorithms or approximation schemes for finding the optimal or near-optimal assortment under a cardinality constraint (see Section \ref{sec:previous}). We show that an assortment that is optimal when customers visit all locations may not be a good assortment of products for other (non-trivial) browsing distributions. 
More broadly, while the literature on constrained assortment optimization is extensive, it is not immediately apparent if its algorithmic results can be used for the \place\ problem.


\subsection{Main contributions }

Our goal is to design an efficient algorithm for a general setting of the \place\ problem. Prior work on assortment optimization has demonstrated the benefits of using the specific characteristics of a choice model in order to design provably good and efficient algorithms. We focus on developing a general algorithmic framework that 
can generate a provably good placement solution for any given choice model and browsing distribution by utilizing an algorithm for assortment optimization in the same choice model. Such a framework allows us to transform a customized assortment optimization algorithm for a given choice model into an approximation algorithm for \place\ with the same choice model.

\vspace{2mm}
\noindent	{\bf Placement for Maximizing Revenue:} Our main technical contribution is a general approximation algorithm for the \place\ problem when products have 
fixed prices and the objective is to maximize expected revenue. Our result holds for a very general family of choice models, { including models for multi-purchase behavior. Specifically, we only require the choice model to satisfy the weak-rationality and monotonicity properties, a condition met by a large class of single-purchase choice models, including for any RUM-based choice model, and several multi-purchase choice models}. 
Given a choice model that satisfies this condition, any browsing distribution, and an 
oracle (black box algorithm) that provides an $\alpha$-approximate solution to the cardinality constrained assortment optimization problem under the given choice model, we give an algorithm that achieves an approximation ratio of { $\frac{\Theta(\alpha)}{\log \rho}$ for the \place\ problem, where $\rho$ is a problem parameter that depends on the variation in the \emph{number} of locations visited by different customers. We have $\log \rho\approx 1$ if all the customers visit nearly the same number of locations, regardless of the total number of display locations. In the worst case, $\rho=m+1$, where $m$ is the total number of display locations, so we get $\frac{\Theta(\alpha)}{\log m}$-approximation for the \place\ problem. }

Through a family of problem instances,  we show that the approximation factor of $\frac{1}{\log m}$ is tight for our algorithm (up to constants). Our algorithm is randomized and utilizes the oracle for finding $\alpha$-approximate cardinality constrained assortments multiple times. We also investigate a restricted version of our algorithm that uses the oracle for assortment optimization at most once and show that this has a significantly worse approximation guarantee of $\frac{\Theta(\alpha)}{ \sqrt{m}}$. This demonstrates 
the	power of using multiple assortments from the oracle. Finally, we give an efficient deterministic $\frac{\Omega(\alpha)}{ \log m}$   algorithm for the Markov Chain choice model, enhancing the solution's applicability and practicality in that context.



\vspace{2mm}
\noindent
{\bf Placement for Maximizing Sales:} 
Sales maximization is a special case of revenue maximization. Specifically, one can find a placement that optimizes expected sales by solving an instance of the revenue maximization problem with identical prices for all the products. When prices are identical and the {(single-purchase)} choice model fits within the RUM framework, the problem of cardinality constrained assortment optimization reduces to an instance of cardinality constrained submodular maximization for a general class of choice models~(\cite{berbeglia2020assortment}). 
We prove that the more general {\sf Placement} problem for products with uniform/identical prices is equivalent to maximizing a monotone submodular function under a \emph{matroid constraint}. Submodular maximization subject to a matroid constraint is a well-studied problem that has a $(1-1/e)$-approximation algorithm~(\cite{calinescu2011maximizing}). This gives us a $(1-1/e)$-approximate algorithm for \place\ with uniform prices. We also show that this is the best possible approximation guarantee for any polynomial time algorithm (unless P=NP). Our hardness result is based on the hardness of finding the optimal cardinality constrained assortment for a Mixture of MNL choice model when products have identical (or uniform) prices. 

Finally, we conduct numerical experiments to evaluate the performance of our algorithms for the \place\ problem on synthetic data under two browsing distribution scenarios: placement on a line and a two-dimensional grid, and assuming a MNL choice model. To benchmark our algorithms, we design a mixed integer linear programming (MILP) formulation to compute the optimal solution of the \place\ problem for MNL choice model (Appendix \ref{appendix:IP}). Our algorithms are highly efficient, consistently delivering near-optimal solutions across all instances we tested. In contrast, the MILP approach struggles to solve larger instances within practical time limits. These results underscore the scalability and effectiveness of our algorithms in solving the \place\ problem across a variety of settings. 

\color{black}


\subsection{Related Literature}\label{sec:previous}
The impact of product placement has been studied in a variety of fields including operations research, marketing, and computer science. {\color{black} Recall that 
the two-stage process in \place\ is inspired by consider-then-choose choice models~(\cite{consider,jagabathula2024demand}), where the customer first shortlists a subset of the assortment or products and then chooses from the subset according to a discrete choice model. The main difference is that the probability that a product is in the consideration set depends on its placement and in this way our model endogenizes the consideration set of customers.
 } We assume that the browsing distribution and the choice model are known and do not depend on the placement. Our goal is to find the revenue optimizing assortment and product placement. This problem has connections to a variety of formulations considered in the literature.

\vspace{2mm}
\noindent \textbf{Product Framing and Display Optimization:}  
The literature on product framing and display optimization is perhaps the most closely related to our work. 
This literature considers settings where customers browse through the display locations in a fixed direction, such as from top to bottom in a list of ads on an online search platform, and then choose at most one product from the set of products that they see. As we discuss in Section \ref{sec:prelim}, this corresponds to a special case of our setting where the browsing distribution is supported on nested sets of locations. We call this  problem \placel.  To the best of our knowledge, finding efficient algorithms for general browsing distributions remained open prior to our work. { Most of the literature on \placel\ assumes that the browsing distribution and choice model are independent of the placement of products. There are exceptions, such as  \cite{gallego2020approximation} and \cite{cultural}, that we discuss in more detail below.}

\cite{davis2015assortment} introduce the problem of \emph{assortment over time} where retailers build their assortment incrementally. Their setting is a special case of \placel\ where the browsing distribution is uniform. They give a $0.5\alpha$-approximation for the problem for any monotone choice model that admits an $\alpha$-approximation to the cardinality constrained assortment optimization problem. \cite{gallego2020approximation} consider a generalization of this setting where the browsing distribution is non-uniform. They give a $6\alpha/\pi^2$-approximation algorithm when the browsing distribution satisfies the \emph{new better than used in expectation} (NBUE) property. {  They also develop novel approximation algorithms for settings where the choice model either depends on the product placement or depends on the number of locations visited by the customer.} 
\cite{alidany} consider \placel\ with MNL choice model and non-uniform browsing distribution. They give a polynomial time approximation scheme (PTAS) based on approximate dynamic programming. A distinctive feature of their model is that a product cannot be displayed at multiple locations in the list and no location can be left empty. This makes the optimization problem more challenging. \cite{feldman2021display} consider a model with additional ``desirability" constraints on the products displayed at each level and give a PTAS. \cite{asadpour2023sequential} consider \placel\ with additional constraints and develop a $(1-1/e)$-approximation for the setting where all products have the same price.

Beyond \placel, \cite{cultural} considered the problem of placing art in a gallery and proposed a Pathway-MNL (P-MNL) model that captures the browsing behavior of customers as a Markovian random walk over the various locations. In their model: 
(1) visitors are not necessarily making a choice between substitutes but are instead visiting several artworks before leaving the gallery and the goal is to engineer visitors' experience, (2) each product (artwork) can be displayed at no more than one location, and (3) the visitor's browsing distribution can depend on the placement of artwork. They show that the problem of optimizing the placement of artwork in the gallery is NP hard. They propose various heuristics to solve the problem and conduct numerical experiments to demonstrate the efficacy of these heuristics but do not analyze the performance of these heuristics theoretically. In our model, the goal is to engineer customer choice and we assume that the browsing distribution is independent of the placement decision but is otherwise arbitrary (and not necessarily Markovian) and we consider a general choice model (not just MNL). We also allow a product to be displayed at multiple locations (in line with retail practice) and focus on developing efficient algorithms with theoretical performance guarantees.  

\vspace{2mm}
\noindent
{ \textbf{Placement Dependent Choice Models:} 
The effect of placement on customer choice can also be captured without the notion of consideration sets by using a placement dependent choice model. In such models, the utility received by the customer is function of both the product and its location. 

One common approach is to consider a choice model over an expanded universe that consists of all (product, location) pairs. The revenue optimal placement in this model is the optimal solution to an assortment optimization problem with constraints ensuring that every location has at most one product (\cite{abeliuk2016assortment}). This constraint can be expressed using a partition matroid or a system of totally unimodular (TU) constraints. For MNL choice model, \cite{avadhanula2016tightness} and \cite{sumida2021revenue} give polynomial sized LP formulations to solve the TU constrained assortment optimization problem. 
For the Markov choice model (\cite{blanchet2016markov}), \cite{SO} gave a $\frac{1}{4}$-approximation algorithm for assortment optimization subject to a matroid constraint and \cite{desir} showed that assortment optimization with TU constraints is NP hard to approximate with a factor of $O(n^{1/2-\epsilon})$ for any fixed $\epsilon>0$. 

A related stream of literature considers the assortment/placement optimization problem for a generalization of MNL choice model in which the customer is shown a finite sequence of assortments until they choose a product (\cite{flores2019assortment,feldman2022multinomial, gao2021assortment, najafi2024multiproduct}). The sequential process induces a placement dependent choice model 
(over the expanded universe) 
in which the customer chooses a product placed in the $\ell$-th position (assortment) only if they do not like any of the products placed in the first $\ell-1$ positions (assortments). In particular, \cite{gao2021assortment} propose a formulation that combines the considerations sets from \placel\ problem with the multi-stage MNL choice model.}

\vspace{2mm}
\noindent
{ \textbf{Sequential Search:} Our model also bears some resemblance to discrete choice models based on sequential search~(\cite{seq0}). In sequential search, the customer initially does not know the value of the product at any location and searches for the best alternative by examining the locations one by one. 
	{\color{black} There is a cost for examining each location and the search process induces a placement dependent browsing distribution over display locations because the customer's decision to stop searching may depend on the set of products that they have seen. }
	The setting considered by \cite{derakhshan2022product} is perhaps the most closely related to our work. They consider a product placement problem with sequential search where customers examine the display locations in a fixed order (similar to \placel) to form a consideration set and then choose at most one product from this set according to a Mixture of MNL choice model. 
	They give novel approximation algorithms and polynomial time approximation schemes for various types of objectives. {\color{black} In contrast, we consider a general store layout but assume that the search process is fully captured by the browsing distribution and independent of product placement.}  
	We refer to \cite{wang2018impact, seq1, seq2} for more detailed reviews of the literature on assortment optimization and discrete choice models based on sequential search.}

\vspace{2mm}
\noindent
\textbf{Assortment Optimization:}  Recall that the problem of cardinality constrained assortment optimization is a special case of our setting where customers see all products on display. The placement of products does not influence customer choice but the finite number of display locations induces a cardinality constraint. 
There is a long line of work on assortment optimization for {single-purchase choice models.} \cite{rusmevichientong2010dynamic} gave a polynomial time algorithm for the problem for MNL choice model. For the Nested Logit choice model, \cite{gallego2014constrained} showed that the problem can be solved in polynomial time using linear programming. For the Mixed MNL choice model with a constant number of mixtures, there is a fully polynomial time approximation scheme (FPTAS) (\cite{mittal2013general,desir2022capacitated}). \cite{desir} gave a  $0.5-\epsilon$ approximation for the problem (with algorithm runtime proportional to $\frac{1}{\epsilon}$) in the Markov choice model and also showed that the problem is APX-hard. {Cardinality constrained assortment optimization has also been studied for various consider-then-choose choice models with placement independent consideration sets. We refer to \cite{aouad2019click, consider} for a detailed review of these settings.

 Several recent works consider the assortment optimization problem for multi-purchase choice models (\cite{bai2023assortment, tula, jasin2024assortment, chenmulti, abdallah2024multi}). In particular, both \cite{bai2023assortment} and \cite{abdallah2024multi} propose novel multi-purchase choice models that satisfy the weak-rationality assumption (see Assumption \ref{ration} in Section \ref{sec:prelim}) and give algorithms for cardinality constrained assortment optimization.} 


	%

\vspace{2mm}
\noindent	\textbf{Shelf-Space Allocation:} There is a significant body of work on the problem of allocating shelf-space in physical stores that sell a variety of  (not necessarily substitutable) products~(\cite{flamand2016promoting, flamand2018integrated, flamand2023store, hariga2007joint}). To the best of our knowledge, the model and methodologies considered in this related literature and our paper are complementary. 
The common solution methodology in this literature is to formulate a mathematical program (typically a non-linear mixed integer program) that incorporates a variety of important practical considerations in store wide allocation of shelf-space for different product categories that are not substitutable. 
We refer to \cite{flamand2023store} for a more detailed review of work on shelf-space allocation.

\vspace{2mm}
\noindent
\textbf{ Outline.}	
The primary contribution of this paper is to present an algorithmic framework that offers theoretically proven guarantees for solving the {\sf Placement} optimization problem. The rest of the paper is organized as follows. In Section \ref{sec:prelim}, we introduce the mathematical formulation of our model.  As a warm-up, in Section \ref{sec:homo}, we examine the scenario where products are uniformly priced. We study the theoretical complexity of this problem and propose an algorithm with a theoretical performance that  matches the hardness lower bound. In Section \ref{sec:hetero}, we consider the general case of products with heterogeneous prices and present our main algorithmic results. 

\section{Model Formulation}\label{sec:prelim}

\paragraph{Choice model.} Consider a ground set $N$ of $n$ substitutable products with fixed prices $(r_i)_{i\in N}$. Let $\phi$ denote a choice model on the ground set $N$. {We start with classic single-purchase discrete choice models where the customer chooses at most one product.} The option of not selecting any product is symbolically represented as product $0$, referred to as the no-purchase option. When any given assortment $S \subseteq N$ is offered to a customer, the choice model prescribes a probability of $\phi(i,S)$ for picking product $i \in S\cup \{0\}$ as the one to be purchased.  
Let $R(S)$ denote the expected revenue of an assortment $S\subseteq N$ which is given by
\[R(S)=\sum_{i\in S} r_i \phi(i,S).\]
{ We also consider multi-purchase choice models to capture settings where the customer may choose more than one product. For multi-purchase models, we use $\phi(i,S)$ to denote {\color{black}\emph{the expected number of times that the customer purchases product $i$}.} 
Note that the expected revenue is still given by $R(S)=\sum_{i\in S} r_i\, \phi(i,S)$.}

\paragraph{Set of locations and placement.} Let $G=[m]$ denote a set of $m$ locations. Let $2^G$ denote the set of all subsets of locations. {\color{black} We have to place/display at most one product at each location. }As we discuss later (Remarks \ref{fill} and \ref{multiple}), this constraint 
is without loss of generality (w.l.o.g.) in our formulation.  
We have many copies of each product and copies of a product can be placed/displayed at multiple locations in $G$. 	
Let $X:{G}\to {N }$ denote a placement  of products over $G$.   More precisely, for each \(j \in G\), the value \(X(j)\) represents the product placed at location~\(j\). 
For notational convenience and to facilitate a more succinct representation, given a subset \(L \subseteq G\), we use \(X(L)\) to represent the set of (distinct) products placed at the locations specified by \(L\), i.e., $X(L)= \cup_{j \in L} \{X(j)\}$. configuration.

\paragraph{Browsing distribution.}  A customer randomly browses/visits a subset of locations on $G$ (see Figure \ref{fig:browse} for an example) and then chooses from the products placed at the visited locations according to choice model $\phi$. Customers' browsing behavior induces a distribution over $G$, we call this the browsing distribution $B$. Let $\mathbb{P}_{B}(L)$ denote the probability that the customer visits/browses the locations in $L$ in some order and then chooses from the set $X(L)\cup\{0\}$ of products (and no-purchase option).   
{\color{black} 	Let $j_{\min}$ ($\geq 1$) denote the size of the smallest (non-empty) set with non-zero support in the browsing distribution. Similarly, let $j_{\max}$ ($\leq m$)  denote the size of the largest set with non-zero support in $\mathbb{P}_B$.} Let $\rho=\frac{j_{\max}+1}{j_{\min}}$. This parameter captures the variation in the number of locations visited by different customers. 
	Observe that $\rho\leq m+1$ and $\rho$ is independent of $m$ when $j_{\max}$ is within a constant factor of $j_{\min}$. W.l.o.g., let $\rho=\frac{j_{\max}+1}{j_{\min}}=2^d$ for some integer $d\in \{1,2,\cdots,\ceil{\log_2 (m+1)}\}$. This assumption is without loss of generality because we can simply increase $j_{\max}$ until $j_{\max}=2^d\, j_{\min} -1$, for some $d\geq 1$, by introducing dummy locations that are visited with zero probability.

	The browsing distribution may depend on the layout of the store and the customers' patience level and capacity for browsing. For example, locations in high traffic areas may be visited with a higher probability, and customers who are pickier and more patient may spend a longer time in store and browse more locations before choosing.   

\begin{figure}[h]
\centering
\includegraphics[scale=0.35]{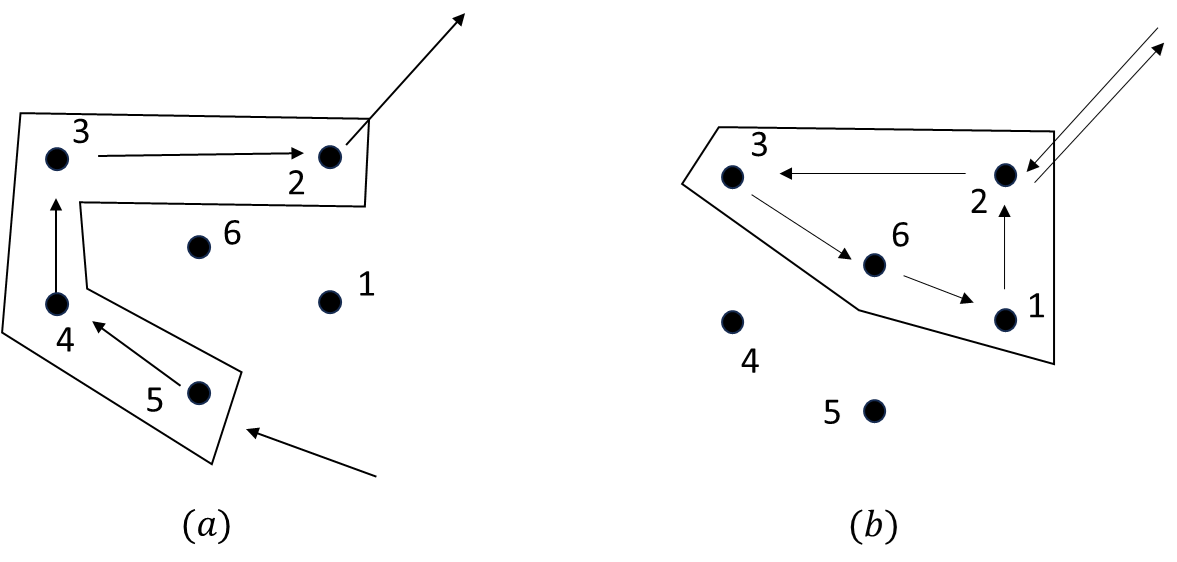}
\caption{Examples of different customer browsing behavior. There are 6 display locations. In $(a)$ the customer looks at products displayed on locations $\{2,3,4,5\}$ and in $(b)$ the customer looks at locations $\{1,2,3,6\}$. The sequence in which customer visits the locations is not important in our model.}
\label{fig:browse}
\end{figure}

\paragraph{Objective.} The expected total revenue of placement $X$ is given by
\[\mathcal{R}(X)=\sum_{L\in 2^{G}} \mathbb{P}_B(L)\,R(X(L)).\]

Let $\mathcal{X}$ denote the set of all placements. Our goal  is to find a placement with the highest expected revenue. We refer to this problem as \place\ and it is given by 
\begin{equation*} \tag{\sf Placement} \label{place}
\qquad\max_{X\in \mathcal{X}}\, \mathcal{R}(X)
\end{equation*}

We make the following two assumptions in our model.

\begin{assumption}\label{ration}\label{weaks}
{Choice model $\phi$ has the following properties:
	\begin{eqnarray*}
		\text{Weak-rationality (or substitutability): }&\quad &\phi(i,S)\geq \phi(i,S\cup\{j\})\quad \forall i\in S, j\not\in S, S\subseteq N.\\
		\text{Monotonicity: }&\quad \sum\limits_{i\in S} & \phi(i,S\cup \{j\})\geq \textstyle \sum\limits_{i\in S}\phi(i,S)\quad \forall j\not\in S, S\subseteq N.
\end{eqnarray*}}
\end{assumption}

{This is a standard assumption in the literature on assortment optimization for single-purchase choice models. The weak-rationality property indicates that the probability of choosing a specific product from a given assortment cannot increase when new products are added to the current assortment. Instead, it can only decrease, as people might potentially opt for the new product as a substitute. The monotonicity property states that in expectation, the number of purchased products increases as we add more products to the assortment. Assumption \ref{ration} is satisfied by all choice models that fall within the RUM framework, which encompasses popular (single-purchase) choice models such as MNL, Markov Chain, Mixture of MNLs, Nested Logit. Assumption \ref{ration} is also satisfied by the multi-purchase choice models proposed by \cite{bai2023assortment} and \cite{abdallah2024multi}. Multi-purchase choice models that capture synergistic effects between complementary products do not satisfy Assumption \ref{ration} (\cite{tula, chenmulti,crosscat}).}
Our second assumption is as follows.

\begin{assumption}\label{decouple}
{
	The browsing distribution $B$ and customer choice model $\phi$ are independent of the placement $X$.} 
\end{assumption}

{
Apart from the independence assumption above, we make no other structural assumptions on the probability distribution $B$, i.e., we allow for all possible customer browsing behavior over the set of $G$. Assumption \ref{decouple} is also ubiquitous in the literature on special cases of \place\ and other closely related problems discussed in Section \ref{sec:previous}~(\cite{davis2015assortment, gallego2020approximation, alidany, asadpour2023sequential}). At the same time, there is  also  literature on models where the browsing distribution or the choice model is placement dependent (see Section \ref{sec:previous}). 	

Overall, we believe that our model is fairly reasonable in settings such as fashion where there are many unfamiliar products and customers may spend some time looking around to shortlist (and try on) a few products before deciding which products to purchase. In a setting such as a grocery store, customers may not be as likely to consider-then-choose and may be more inclined to purchase the first suitable product they come across. In an effort to also capture such settings, in Appendix \ref{appx:decouple} we partially relax Assumption \ref{decouple} and consider a placement dependent MNL model where customers are more likely to purchase products that 
they saw earlier on their path. We show that our main algorithm and its performance guarantee generalize to this setting.}

In order to build our optimization framework for \place, we assume access to two oracles. The first oracle computes an optimal or approximate solution for the cardinality {constrained} assortment optimization problem under the choice model $\phi$. The second oracle outputs any number of independent samples from the browsing distribution. In the following sections, we formally present the two oracles.
{\color{black}
\begin{remark}[Product Repetition] Placing a product in multiple locations is common in both physical and digital retail environments.  In grocery stores, for instance, a product might appear on end-cap displays, checkout lanes, or promotional aisles. Big-box retailers often use promotional zones or pallet displays to emphasize high-volume or discounted items. Apparel stores highlight featured styles using mannequins and dedicated displays. Similarly, e-commerce platforms showcase repeated products in “best-selling” or “featured” sections, often at the top of the page. As illustrated in Figure \ref{fig:one}, a single product search may return multiple listings for the same base item, differentiated by size, color, or configuration.
\end{remark}}
\begin{remark}[Filling empty locations with highest price product] \label{fill} {\color{black} Since we are allowed to repeat a product at multiple locations, we can ensure that all locations are filled without losing revenue. }
    Suppose we are given a 
    placement solution where some locations are left empty. We describe a simple process to transform a placement with empty locations into a feasible assortment without decreasing the expected revenue of the placement. Let $i^*$ denote the product with the highest price, i.e., 
\[i^*=\argmax_{i\in N} r_i.\]
Using the weak-rationality {and monotonicity of $\phi$ (Assumption \ref{ration})}, we have the following well known inequality, 
\[R(S\cup\{i^*\})\geq R(S) \quad \forall S\subseteq N\backslash\{i^*\}.\]
Let $Y$ denote a placement where some locations are empty. Let $\hat{Y}$ denote a placement that is identical to $Y$ except that $i^*$ is placed in every location that is empty in $Y$. We have,
\[\mathcal{R}(\hat{Y})=\sum_{L\in 2^G} \mathbb{P}_B(L)R(\hat{Y}(L))\geq \sum_{L\in 2^G} \mathbb{P}_B(L)R(Y(L))=\mathcal{R}(Y).\]
From the above, the optimal solution to \place\ has the same revenue as the optimal solution to \place\ with the additional constraint that all locations must be filled. 
Thus,  if $Y$ is a $\beta$-approximate solution to \place\ 
then $\hat{Y}$ is a $\beta$-approximate solution to the constrained problem as well. 
\end{remark}
\begin{remark}[Multiple products in a location] \label{multiple} In practice, there may be settings where more than one product can be  simultaneously displayed at a given display location. This setting is referred to as \emph{product framing} in related work (\cite{gallego2020approximation}). We can capture this in our model by using a set of locations to represent each real/physical location and defining a browsing distribution such that a customer either visits the whole set of locations (that represent the real display location) or none of the locations in the set.
	\end{remark}


\begin{remark}[Placement on a Line.]\label{pline} In some settings, the browsing patterns are highly structured and the browsing distribution may be represented explicitly using a small number of parameters. For example, in an online recommendation platform or a search engine, products are often placed vertically on a web page and customers visit a subset of consecutive locations starting from the topmost location~(\cite{alidany}). In this setting the browsing distribution is defined using $m$ parameters that represent the probability of customer visiting locations $1$ (topmost location) through $j$  ($\leq m$) (See Figure \ref{fig:line}). This setting will be useful for motivating many of our ideas in the subsequent discussion and we refer to it as \placel.


\begin{figure}[h]
	\centering
	\includegraphics[scale=0.35]{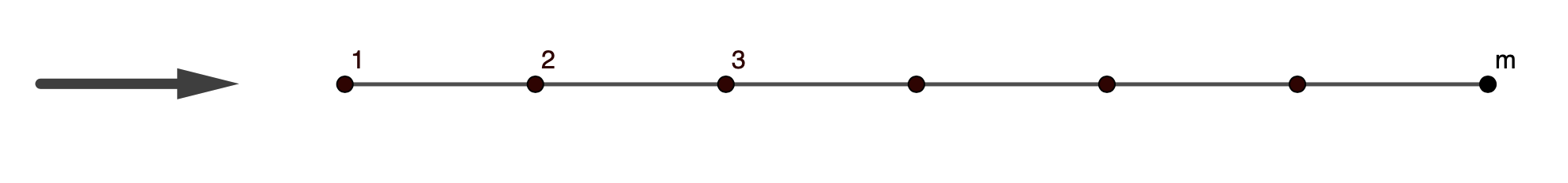}
	\caption{Example of \placel}
	\label{fig:line}
\end{figure}
  {\color{black}  For \placel, there is no advantage in repeating a product:  If a product is at location $j$, then a copy of the product at location $j+k$ does not contribute additional revenue because a customer who visits location $j+k$ will also visit location $j$. So if a product is repeated at $j$ and $j+k$ then we can simply leave $j+k$ empty without changing revenue. In Section \ref{subsec:general}, we show that repeating products makes a big difference to the revenue for the general \place\ problem.}
\end{remark}


\subsection{Oracle for Cardinality Constrained Assortment Optimization} \label{sec:oracle}

Consider the following browsing distribution,
\begin{equation*}
\mathbb{P}_B(L)=
\begin{cases}
	1 & \quad L=G,\\
	0& \quad \forall L\subset G. 
\end{cases}
\end{equation*}
In this case, the expected revenue of a placement $X$ reduces to the expected revenue of the assortment of all products placed over $G$, i.e., $\mathcal{R}(X)=R(X(G))$. Thus, \place\ is a (strict) generalization of cardinality constrained assortment optimization. Recall that our goal is to give a general algorithm for the placement problem that uses existing algorithms for cardinality constrained assortment optimization as a subroutine to find a provably good placement of products. We assume access to an oracle that outputs an $\alpha$-approximate solution to the cardinality constrained assortment optimization problem, i.e., given a cardinality parameter $k$, the oracle gives a set $S^*_k$ such that
\begin{equation} \label{eq:approx}
R(S^*_k)\geq \alpha \,\, \max_{|S|\leq k}R(S).
\end{equation}

Our algorithms will use this oracle as a black box and the approximation guarantee of our algorithm will be proportional to $\alpha$. { As we discussed in Section \ref{sec:previous}, there is an extensive literature cardinality constrained assortment optimization under various single-purchase choice models. 
 \cite{bai2023assortment} introduce a multi-purchase choice model that satisfies Assumption \ref{ration} and give a PTAS for cardinality constrained assortment optimization. \cite{abdallah2024multi} introduce a multi-purchase model based on the RUM framework and give a near-optimal algorithm for cardinality constrained assortment optimization in an asymptotic regime (large number of products).}


Without loss of generality (w.l.o.g.), let 
\[|S^*_k|=k\quad \forall k\in [m].\]
To see this, observe that we can consider an expanded ground set that includes several dummy products with purchase probability 0 in every assortment.  
Adding these dummy products to an assortment does not change its revenue. On the expanded ground set we have, $|S^*_k|=k\quad \forall k\in [m].$
Given a placement solution over the expanded ground set, we can obtain a placement solution over the original ground set $N$ by replacing the dummy products with copies of product $i^*$ (the product with the highest price in $N$). Recall that adding $i^*$ does not decrease the expected revenue of any placement solution.






\subsection{Sampling Oracle for Browsing Distribution}\label{sec:sample}
For the sake of simplicity, 
we state and analyze the approximation guarantee of our algorithms assuming that we have \emph{exact knowledge of the expected revenue of any placement}. In practice, the expectation may be hard to compute as the support of the browsing distribution can be as large as $2^m$, where $m$ is the number of locations. 
Our algorithms and their approximation guarantees hold even when we only have access to an oracle that outputs independent samples from the browsing distribution. 
 We can use estimates of the expected revenue obtained by generating a small number (quantified later) of independent random samples of visited locations using a sampling oracle for the browsing distribution.





Using the sampling oracle, we can obtain a good approximation of $\mathcal{R}(X)$ for any given placement $X$ using the Monte-Carlo method. Let \opt\ denote the expected total revenue of the optimal solution.

\begin{lemma}\label{est}
Given parameters $(\epsilon,\delta)\in[0,1]^2$ and a placement $X$, the sample average revenue of $\frac{m^2\log \delta^{-1}}{2\epsilon^2}$ independent samples from the sampling oracle gives an estimate $\hat{\mathcal{R}}(X)$ of the expected total revenue such that with probability $1-2\delta$,
\[\hat{\mathcal{R}}(X)\in [\mathcal{R}(X)-\epsilon\,\opt, \mathcal{R}(X)+\epsilon\,\opt].\] 
\end{lemma}
\begin{corollary} \label{coro}
Given a set of placements $\{X_1,X_2,\cdots, X_t\}$, using $\frac{m^2\log \delta^{-1} }{2\epsilon^2}  \cdot \log t$ samples from the sampling oracle we can obtain revenue estimates $\{\hat{\mathcal{R}}(X_1),\cdots, \hat{\mathcal{R}}(X_t)\}$ such that,
\[|\hat{\mathcal{R}}(X_k)-\mathcal{R}(X_k)|\leq \epsilon\,\opt\quad \forall k\in[t],\]
with probability at least $1-2\delta$.
\end{corollary}
We give the proof of Lemma \ref{est} in Appendix \ref{appx:linehard}. The corollary follows by using Lemma \ref{est} and applying the union bound. The next lemma demonstrates the usefulness of measuring estimation error as a factor of \opt. The proof of this lemma is also included in Appendix \ref{appx:linehard}.
\begin{lemma}\label{err}
Consider a set of placements $\{X_1,X_2,\cdots, X_t\}$ with estimates $|\hat{\mathcal{R}}(X_k)-\mathcal{R}(X_k)|\leq \epsilon\,\opt\,\, \forall k\in[t]$. Let $\hat{X}^*$ denote the best placement according to $\hat{\mathcal{R}}(\cdot)$, i.e., $\hat{X}^*=\argmax_{k\in[t]}\hat{\mathcal{R}}(X_k)$. Similarly, let $X^*=\argmax_{k\in [t]} \mathcal{R}(X_k)$. Suppose that $\mathcal{R}(X^*)\geq \beta \opt$. Then, we have
\[\mathcal{R}(\hat{X}^*)\geq \mathcal{R}(X^*) -2\epsilon\, \opt\geq  (\beta-2\epsilon) \opt. \]
\end{lemma}

\begin{remark}
    [Usage of Exact Revenue Oracle]
    Recall that the approximation guarantee of our algorithms is stated under the assumption that we have \emph{exact knowledge of the revenue of any placement}. Relaxing this assumption leads to a small loss in the approximation guarantee due to estimation errors. In particular, the final step in our main algorithm is to find the placement that maximizes $\hat{\mathcal{R}}(\cdot)$ from a small set of candidate placement solutions.

    From Lemma \ref{err}, given an approximation guarantee of $\beta$, our approximation guarantee in the absence of an exact revenue oracle is simply $\beta-2\epsilon$. In special cases such as when the product prices are identical (Section \ref{sec:homo}) or when the choice model is Markovian (Section \ref{subsec:deran}), we propose greedy algorithms that construct a placement solution over several iterations. 
As we discuss later, the error due to estimates of revenue is sufficiently small even in these cases.
\end{remark}	 

\begin{remark}[Deterministic vs. Randomized Algorithms]
    Given an exact revenue oracle, we say that an algorithm is \emph{deterministic} if it does not use any randomness. In contrast, a randomized algorithm uses additional randomness and may output different solutions each time the algorithm is executed.
\end{remark}  
{\color{black}
\begin{remark}[Constructing the Sampling Oracle from Data]
In practice, one may use customer traffic data collected via RFID \citep{larson2005exploratory} and other technologies \citep{heatmap} in physical stores, and browsing data in online stores, to construct a sampling oracle for the browsing distribution. We refer to \cite{larson2005exploratory} and \cite{heatmap} for more details on how these technologies are used in physical stores.  \cite{flamand2023store} propose a different approach based on learning a parametric model of customers' browsing behavior using sales transaction data. The idea of learning something about the customer's path in the store from their set of purchased products is based on the assumption that customers shop in a way that minimizes their walking distance \citep{hui2009research}. Under this assumption, given the store layout, the product placement, and set of products purchased by a customer, one can use optimization techniques to infer the customer's path through the store. \cite{flamand2023store} use this technique to learn the parameters of a model that predicts the probabilities of a customer visiting a product in a given location. For a fixed product placement our model reduces to a consider-then-choose choice model and one can estimate the revenue by using techniques for estimating consider-then-choose models \citep{jagabathula2024demand}. 
\end{remark}
}

\section{Uniformly Priced Products}\label{sec:homo}

In this section, {\color{black} we consider the special case of \place\ where products have the same (uniform) price,} i.e., $r_i = 1\; \forall i \in N$. This captures the case where the objective function is to maximize the sales of products rather than the expected revenue. { This section serves as a warm-up before we present our main technical results for general prices in the subsequent section.
 
In this section, we focus on single-purchase choice models that belong to the general framework of Random Utility Maximization (RUM)} and 
we tackle the \place\ problem with uniformly priced products. In particular,   we present a $(1-1/e)$ approximation algorithm for the problem in Section \ref{sec:unif}. In Section \ref{sec:hard}, we show that this approximation is tight; that is, the \place\ problem with uniformly priced products cannot be approximated within a factor better than $(1- \frac{1}{e}+\epsilon)$ for any $\epsilon >0$, unless P=NP. The hardness result follows from the hardness of cardinality constrained assortment optimization under a (single-purchase) Mixture of MNL (MMNL) model when products have uniform prices.

\subsection{Placement of Uniformly Priced Products}\label{sec:unif}

For uniformly priced products, we demonstrate that the \place\ problem can be transformed into an instance of constrained maximization of a monotone submodular function. The foundation of our reduction is based on the established lemma below.

\begin{lemma} [\cite{berbeglia2020assortment}]
\label{submod}
For any RUM-based choice model, the expected revenue function $R(\cdot)$ is both monotone and submodular when all products are priced identically.
\end{lemma}
Lemma \ref{submod} is a well-established result observed across various RUM-based choice models. As an example, \cite{desir} demonstrates this result for the Markov Chain choice model. \cite{berbeglia2020assortment} provides a proof of this lemma for general RUM-based choice models (see Lemma 4.2 in their paper). Note that there are choice models that are not RUM-based but adhere to Assumption \ref{ration}, in which the revenue function is not submodular, and therefore do not verify the above lemma. \cite{berbeglia2020assortment} presents such an example in their paper. Hence, we restrict ourselves to RUM-based choice models in this section.

Building on Lemma \ref{submod}, we establish that the \place\ problem with identical prices, can be formulated as a problem of maximizing a monotone submodular function under a matroid constraint\footnote{A matroid on ground set $U$ is given by a family of subsets of $U$, denoted $\mathcal{F}$. Every set in $\mathcal{F}$ is called an independent set of the matroid. $\mathcal{F}$ satisfies the following properties: (i) The empty set, $\phi$, is independent. (ii) If $A\in \mathcal{F}$ then $B\in\mathcal{F}$ for all $B\subseteq A$. (iii) If $A,B\in\mathcal{F}$ and $|B|<|A|$, then $ B\cup\{x\}\in \mathcal{F}$ for some $x\in A\backslash B$.}. There are several algorithms for maximizing a monotone submodular function subject to a matroid constraint. \cite{nemhauser1978analysis}  showed that a greedy algorithm that iteratively builds the solution by adding elements one by one 
is $0.5$ approximate. \cite{calinescu2011maximizing} give $(1-1/e)$-approximation algorithms for the problem which is the best known approximation. We refer to \cite{badanidiyuru2014fast} for a review of other algorithms for the problem. The next theorem shows how to obtain $(1-1/e)$-approximation algorithm  for the \place\ problem with identical prices.

\begin{theorem}\label{uniform}
There is a $(1-1/e)$-approximation algorithm for the \place\ problem  with uniformly priced products and a RUM-based choice model.
\end{theorem}

{We include the proof of Theorem \ref{uniform} in Appendix \ref{appx:homo}.} Note that \cite{asadpour2023sequential} show a $(1-1/e)$ guarantee for a similar problem setting by generalizing ideas from submodular optimization. The browsing distribution in their setting is the same as \placel\ but they consider additional constraints arising out of fairness considerations. Our result is for a general browsing distribution without the latter constraints.

An extension of Theorem \ref{uniform}, gives us the following approximation guarantee for the general setting of heterogeneous prices. Let $r_{\max}=\max_{i\in N} r_i$ and $r_{\min}=\min_{i\in N} r_i$. W.l.o.g., $r_{\min}>0$ as all products with price 0 can be discarded from the ground set. 

\begin{corollary} \label{coro:one}
There is a $\frac{O(1)}{\log \frac{r_{\max}}{r_{\min}}}$ approximation algorithm for \place\ with arbitrary prices under a RUM-based choice model.
\end{corollary}

{The proof of this corollary is also included in Appendix \ref{appx:homo}.}

\subsection{Hardness of Placement}\label{sec:hard}

Recall that the cardinality constrained assortment optimization problem is a special case of \place\ (Section \ref{sec:oracle}). In this subsection, we argue that, unless P=NP, a $(1-1/e)$ approximation is the best achievable guarantee in polynomial time for cardinality constrained assortment optimization under a Mixture of Multinomial Logit choice model  (MMNL)  and uniformly priced products. This proves that the $(1-1/e)$-approximation guarantee for the \place\ problem with uniformly priced products (derived in the previous section) is best possible. 


\begin{theorem} \label{thm:hardness}
Unless P=NP, there is no polynomial-time algorithm that approximates the problem of cardinality constrained assortment optimization better than than $(1-\frac{1}{e}+ \epsilon)$ for any $\epsilon>0$  when the choice model is given by a MMNL choice model and all the products have the same price. 
\end{theorem}
The proof of this theorem is included in Appendix \ref{appx:homo}.

\section{Heterogeneous Prices}\label{sec:hetero}

In this section, we present our primary technical contribution: a general approximation algorithm for the \place\ problem  with heterogeneous prices. {Our result holds for any choice model that satisfies the weak-rationality and monotonicity property (Assumption \ref{ration}), including multi-purchase models, and is valid under any placement independent browsing distribution (Assumption \ref{decouple}).} Specifically, given an oracle that offers an $\alpha$-approximate solution to the cardinality-constrained assortment optimization problem under a specified choice model, we propose a randomized algorithm that achieves an approximation ratio of $\frac{\Theta(\alpha)}{\log \rho}$ for the \place\ problem, where $\rho\leq m+1$ and $m$ is the number of display locations. In Appendix \ref{appx:decouple}, we show that this result also holds for a placement dependent MNL choice model.

{\color{black} 
	For heterogeneous prices, we show the following hardness result for \placel.  
	\begin{lemma} \label{lem:hardness}
     \placel\ is weakly NP-hard\footnote{A problem is weakly NP-hard if it is NP-hard, but a pseudo-polynomial time algorithm may exist for its solution.} even when the choice model $\phi$ is MNL.
\end{lemma} We show this using the hardness of product framing problem  shown by \cite{gallego2020approximation}. The proof is included in Appendix \ref{appx:hetero}. 
Thus, unless P=NP, it is not possible to obtain a polynomial time algorithm for solving the \place\ problem exactly even on a line and for the MNL choice model.
}

The subsequent discussion in this section is structured as follows: Section \ref{subsec:warmup} serves as a preliminary exploration, where we analyze the special case of \textit{\placel} and motivate several of our algorithmic ideas. In Section \ref{subsec:general}, we present  the main result which is our general randomized algorithm for \place. Section \ref{sec:tight} demonstrates the tightness of our approximation guarantee. Lastly, a deterministic algorithm for the Markov Chain choice model is presented in Section \ref{subsec:deran}.


\subsection{Warm-up: Placement on a Line} \label{subsec:warmup}

Recall that	for \placel, see Figure \ref{fig:line}, the browsing distribution is defined using $m$ parameters that represent the probability of customer visiting locations $1$ (topmost location) through $j$  ($\leq m$). Let $(1:j)$ denote the set of locations from $1$ through $j$. Let $\theta_j$ denote the probability that customer visits locations $1$ through $j$, i.e.,
\[\theta_j=\mathbb{P}_B(1:j).\]
The expected revenue of a placement $X$ on a line is simply,
\[\mathcal{R}(X)=\sum_{j=1}^m \theta_j\, R(X(1:j)).\]
Recall that our goal is to give a provably good placement by using algorithms for constrained assortment optimization. In particular, given $\alpha$-approximate algorithm for cardinality constrained assortment optimization, our goal is to give an algorithm that generates a placement with revenue at least $\frac{\Theta(\alpha)}{\log m}$ times the revenue of the optimal placement. For each $k \in [m]$, recall $S_k^*$ is the $\alpha$-approximate assortment of size $k$ given by the oracle in Section \ref{sec:oracle}. First, we start by showing the limitation of restricting the solution of \place\ to simply use  products from $S_k^*$ for some fixed $k$.

%

\vspace{2mm}
\noindent
{\bf Limitations of restricting product selection to a Single  ${\bf S_k^*}$}. 
Fix $k \in [m]$ and let $S_k^*$ be the assortment of size $k$ returned by the cardinality-constrained assortment optimization oracle. We begin by examining  the simple solution that relies solely on products in $S_k^*$ for the \placel\ problem. Let \opt\ represent the optimal value of the \placel\ problem. We demonstrate that the performance of this solution can be substantially suboptimal in the worst-case scenario. Specifically, the subsequent lemma reveals that for any $k \in [m]$, there exists an instance of the \placel\ problem such that any placement of $S_k^*$ yields an expected revenue of no more than $\frac{O(1)}{\sqrt{m}} \cdot \opt$. This holds even if $S_k^*$ is the optimal cardinality constrained assortment, i.e., even if we have an oracle with $\alpha=1$.

\begin{lemma}\label{single}
For any $k \in [m]$, any placement of $S^*_k$ is at most $O(1) \cdot \min (\frac{1}{k}, \frac{k}{m}) \cdot \opt$ in the worst case even for the MNL choice model. Hence, it is at most a $\frac{O(1)}{\sqrt{m}} \cdot \opt$ approximation.
\end{lemma}

The proof of Lemma \ref{single} is included in Appendix \ref{appx:hetero}. From Lemma \ref{single}, we observe that selecting products solely from a fixed $S_k^*$ can be highly restrictive. In the worst case, the expected revenue from such a placement does not exceed $ \frac{O(1)}{\sqrt{m}} \cdot \opt$. Motivated by this lemma, we explore algorithms that, for each $k \in [m]$, compute a placement using products from  $S_k^*$  and then select the best among these placements. We demonstrate that this approach markedly enhances the solution to the \placel\ problem. 

\vspace{2mm}
\noindent
{\bf Importance of product selection from the Best of Many $\bf S_k^*$.} For each $k\in[m]$, let $X^*_k$ represent the placement that positions products from $S^*_k$ in the first $k$ locations of the line (in any given order) while leaving the subsequent locations empty. 
Define $Y$ as the most optimal placement from the set $\{X^*_k\}_{k\in[m]}$. Specifically, consider the placement solution for the problem given by:	
\[{\sf Best\, of\, Many:} \qquad Y= {\argmax_{X\in\{X^*_{k}\}_{k\in[m]}} \mathcal{R}(X).}\]
In the next lemma, we establish that $Y$ gives $\frac{\alpha}{\log \rho}$ approximation to \placel. The result holds for all choice models satisfying Assumption \ref{ration}, including multi-purchase models. Recall from Remark \ref{fill}, that we can always fill empty locations with the highest price product without decreasing revenue. 

{Recall that, $j_{\min}$ ($\geq 1$) and $j_{\max}$ denote the size of the smallest (non-empty) set and largest set (respectively) with non-zero support in the browsing distribution. Also, $\rho=\frac{j_{\max}+1}{j_{\min}}=2^d,$ for some integer $d\in \{1,2,\cdots,\ceil{\log_2 (m+1)}\}$.} 
\begin{lemma}\label{line}
The placement $Y$ gives a $\frac{\alpha}{{\log_2 (\rho)}}$-approximation for \placel.
\end{lemma}
\begin{proof}{Proof.}

Let $O: G \to N$ denote an optimal placement for \placel\ and let $\mathcal{R}(O)=\opt$. {Recall that $O(1:k)$ denotes the set of products placed in locations $\{1,\cdots,k\}$. For $\ell\geq 1$, let $\mathcal{R}_{\ell}(O)$ denote the revenue contribution from products in location $2^{\ell-1}j_{\min}$ to $2^{\ell}j_{\min}-1$, i.e.,
	\begin{equation} \label{eq:Wdef}
		\mathcal{R}_{\ell}(O)= \sum_{j=2^{\ell-1} j_{\min}}^{2^{\ell}j_{\min}-1}\left[r_{O(j)}\,\sum_{k\geq j} \theta_k \,\phi\left(O(j), O(1:k)\right)\right].
	\end{equation}
	Given that $\rho$ is a power of 2, we have 
	\[\mathcal{R}(O)=\sum_{\ell=1}^{d}\mathcal{R}_{\ell}(O).\]
	We have {\color{black} for any $\ell\in[d]$},
	$(2^{\ell-1}j_{\min}:\min\{k,\,2^{\ell}j_{\min}-1\}) \subseteq (j_{\min}:j_{\max}).$ 
	To simplify notation, denote the subset of locations 
	\[G_{\ell,k}= (2^{\ell-1}j_{\min}:\min\{k,\,2^{\ell}j_{\min}-1\}) .\]
	We have $G_{\ell,k} \subseteq (2^{\ell-1}j_{\min}:2^{\ell}j_{\min}-1)$. Therefore, by using \eqref{eq:Wdef} and the weak-rationality of $\phi$ (Assumption~\ref{ration}), we get
	\[\mathcal{R}_{\ell}(O)\leq \sum_{j=2^{\ell-1}j_{\min}}^{2^{\ell}j_{\min}-1}r_{O(j)} \sum_{k=j}^{j_{\max}} \theta_k \phi\left(O(j), O(G_{\ell,k})\right).\]
	Observe that the RHS is the expected total revenue of the placement solution where we place products $O(2^{\ell-1}j_{\min}:2^{\ell}j_{\min}-1)$ between $2^{\ell-1}j_{\min}$ and $2^\ell j_{\min} -1$, same as $O$, and keep all other positions vacant. Let $O_{\ell}$ denote this solution. 
	We have
	\[\mathcal{R}(O)=\sum_{\ell=1}^{d}\mathcal{R}_{\ell}(O)\leq \sum_{\ell=1}^{d} \mathcal{R}(O_\ell).\]
	Consider some $\ell$ in $[d]$. The total revenue of placement $X^*_{2^{\ell-1} j_{\min}}$ can be lower bounded by the revenue of placement $O_\ell$ as follows,
	\begin{eqnarray*}
		\mathcal{R}(O_\ell)&= &\sum_{j=2^{\ell-1}j_{\min}}^{2^{\ell}j_{\min}-1}r_{O(j)} \sum_{k=j}^{j_{\max}} \theta_k \phi\left(O(j), O(G_{\ell,k}) \right),\\
		&= & \sum_{k= 2^{\ell-1}j_{\min}}^{j_{\max}} \theta_k\, R(O(G_{\ell,k})),\\
		&\leq& \frac{1}{\alpha}\sum_{k= 2^{\ell-1}j_{\min}}^{j_{\max}} \theta_k\, R\left(S^*_{2^{\ell-1}j_{\min}}\right),\\
		&\leq&\frac{1}{\alpha}\mathcal{R}(X^*_{2^{\ell-1}j_{\min}}),
	\end{eqnarray*}
	where the first inequality follows from \eqref{eq:approx} and the fact that $|O( G_{\ell,k})| \leq  |G_{\ell,k}| \leq 2^{\ell-1}j_{\min}$.
	Thus,
	\[\mathcal{R}(O)\leq \frac{1}{\alpha}\sum_{\ell=1}^{d}\mathcal{R}(X^*_{2^{\ell-1}j_{\min}})\leq \frac{\log_2 (\rho)}{\alpha}\, \mathcal{R}(Y),\]
}here the last inequality follows from the fact that $\mathcal{R}(Y)\geq \mathcal{R}(X^*_k)\,\, \forall k\in[m]$ and $d=\log_2(\rho)$.
\hfill\Halmos		\end{proof}

\begin{remark}
    Recall that in the proof of Lemma \ref{line}, we make use of at most 
$\log_2 \rho$ candidate placements. In Section \ref{sec:tight}, we show that our analysis is tight even if we take the best placement out of the $m$ candidates $\{X^*_k\}_{k\in[m]}$.
Different methods for placing/ordering the products from sets $S^*_k$ do not improve the worst case guarantee.
\end{remark}

\subsection{Main Results: Approximation Guarantee} \label{subsec:general}

Given an assortment, finding the optimal placement for general browsing distributions is significantly more involved as customers may start browsing from various locations and may also travel between locations in myriad ways. In particular, the optimal solution may now require copies of a product to be placed in multiple locations. We illustrate the importance of repeating products with a simple example. {\color{black} Recall that, for \placel\ the optimal placement without product repetition has the same revenue as the optimal placement with product repetition (see Remark \ref{pline}).}
{
\begin{example}\label{ex1}
	Consider the scenario where customer starts at each location $j\in G$ with probability\ $\frac{1}{m}$ and then either chooses the product at location $j$ or departs, i.e., $\mathbb{P}_B(\{j\})=\frac{1}{m}\,\, \forall j\in G$.
	Let $i_1=\argmax_{i\in N} r_i\phi(i,\{i\})$. Since we can repeat a product at multiple locations, the optimal solution is to place product $i_1$ at all locations. This placement has expected revenue $R(\{i_1\})$. Now, 
    suppose that all products apart from $i_1$ have price 0. 
    Then, the expected revenue of the optimal placement that does not repeat any product is at most $\frac{1}{m} R(\{i_1\})$.  
\end{example}
}
 This sheds light on a new trade off for the general placement problem - between displaying a larger collection of products and repeating the same (high revenue) product(s) at multiple locations.  We address this new difficulty by using randomness in the placement.

\textbf{Algorithm:} Given assortment $S^*_k$ as defined in Equation \eqref{eq:approx}, let $X^r_k$ denote a uniformly random placement of products in $S^*_k$, i.e., at each location we place a uniformly randomly (and interdependently) chosen product from $S^*_k$. 
Note that a single product in $S^*_k$ may be placed at multiple locations (we sample with replacement). Let, 
\[Y=\argmax_{X\in\{X^r_k\}_{k \in [m]}} \mathcal{R}(X),\]
{\color{black} denote the final (randomized) placement solution. We can compute $Y$ by finding the highest revenue placement from the set $\{X^r_k\}_{k\in[m]}$ of $m$ candidate solutions. 
 For any given $k \in [m]$, the random placement $X^r_k$ is oblivious to how customers travel between locations. In contrast, for the \placel\ problem, we greedily placed the products in $S^*_k$ at the top $k$ spots. 
The oblivious nature of randomized placement is pivotal in proving analytical guarantees.

We use a method akin to the probabilistic method to prove a guarantee on the expected performance of the algorithm. Specifically, we will show a lower bound on the expected value $\ep[\mathcal{R}(Y)]$ of the randomized solution $Y$.  Standard repetition can be used to obtain an algorithm that generates a good placement with high probability.  Later, we also give a derandomized algorithm for the Markov chain choice model.} 	
Let \opt\ represent the optimal value of the \place\ problem.
Recall that, $\rho=\frac{j_{\max}+1}{j_{\min}}=2^d$ for some integer $d\in \{1,2,\cdots,{\log_2 (m+1)}\}$.
We establish the following performance guarantee for our algorithm.

\begin{theorem}\label{randgen}
The randomized placement $Y$ gives $\frac{\Omega(\alpha)}{\log \rho}$-approximation for \place, i.e.,  
\[\ep[\mathcal{R}(Y)]=\ep[\max_{k \in [m]} \mathcal{R}(X^r_k)]\geq \alpha\frac{1-1/e}{{\log_2 \rho}} \opt.\]	
\end{theorem} 
{ Theorem \ref{randgen} holds for all choice models that satisfy Assumption \ref{ration}, including multi-purchase models. In Appendix \ref{appx:decouple}, we extend this result to a placement dependent MNL choice model. } 

To prove Theorem \ref{randgen}, we first show the result for a line graph (Lemma \ref{randline}). We then use the result for a line graph as a building block to establish Theorem \ref{randgen}. {\color{black}	We start by describing the main ideas at a high level. 
\smallskip

\noindent \emph {Proving the result for a line:} Recall that in Lemma \ref{line} we showed that the best (deterministic) placement from the set $\{X^*_k\}_{k\in[m]}$, is at least $\frac{\alpha}{{\log_2 \rho}}$-approximate. Also, recall that the exact order of placement had no impact on the result, i.e. $X^*_k$ can be any placement where the products in $S^*_k$ are placed in the first $k$ spots on the line. In the new algorithm, we allow repetition of products and the candidate placement $X^r_k$ may not place some products in $S^*_k$ at any location. In fact, the value of $\mathcal{R}(X^r_k)$ may be very low with some probability, such as, when a ``low revenue" product is placed at every location. In Lemma \ref{helper}, we show that such possibilities have a limited effect on the expected revenue. Then, in Lemma \ref{randline}, we show that the sum of the expected revenues of the $\log_2(\rho)$ candidate placements $\{X^r_{j_{\min}}, X^r_{2j_{\min}},\cdots,X^r_{2^{d-1} j_{\min}}\}$, is within a constant factor of the optimal solution. }

\begin{lemma}\label{helper}
Given the assortment $S^*_j$ for $j\in[m]$, let $S^r_j$ denote  a random assortment of $j$ products where each product is independently sampled from $S^*_j$ (with replacement). We have,
\[			\ep[R(S^r_j)] \geq \left(1-\left(1-\frac{1}{j}\right)^j\right)R(S^*_j)\geq (1-1/e) R(S^*_j).\]
\end{lemma}

\begin{lemma}\label{randline}
For a line graph,	$
\sum_{\ell=1}^{\log_2(\rho)}\ep[\mathcal{R}(X^r_{2^{\ell-1}j_{\min}})]\geq \alpha(1-1/e) \opt$.
\end{lemma}

 Using the fact that, 
\[\sum_{\ell=1}^{\log_2(\rho)}\ep[\mathcal{R}(X^r_{2^{\ell-1}j_{\min}})]\leq (\log_2 \rho)E[\mathcal{R}(Y)],\] 
we have that $Y$ is a $\alpha \frac{(1-1/e)}{\log_2 \rho}$-approximation for a line graph as a direct corollary of Lemma \ref{randline}. The proofs of Lemma \ref{helper} and \ref{randline} are included in Appendix \ref{appx:hetero}.

{\color{black}
\noindent \emph {From a line to the general case:} At a high level, our approach is based on the fact that the browsing distribution of a general instance can be 'decomposed' into a weighted 'sum' of browsing distributions on line graphs. To make this precise, we introduce an equivalent (though perhaps less intuitive) characterization of the browsing distribution. Let $\pi:[m]\to[m]$ denote a permutation over $[m]$ and let $\Pi$ denote the set of all such permutations. For a permutation $\pi$, define $\theta^{\pi}_j$ as the probability that customer visits locations $\{\pi(1), \pi(2),\cdots, \pi(j)\}$ exactly in this order\footnote{Let $L=\{\pi(1), \pi(2),\cdots, \pi(j)\}$. Recall that, $\mathbb{P}_B(L)$ represents the probability that $L$ is the set of locations visited by the customer in \emph{any order}. In contrast, $\theta^{\pi}_j$ is the probability that the customer visits the locations in the specified order.}. Let $X^\pi(1:j)$ denote the set of products placed at $\{\pi(1),\cdots,\pi(j)\}$. Then the total expected revenue can be expressed as a sum over all permutations,
\begin{equation} \label{decompooo}
\mathcal{R}(X)=\sum_{\pi\in \Pi} \sum_{j\in[m]} \theta^{\pi}_jR(X^{\pi}(1:j)).
\end{equation}
Now fix a permutation $\pi\in \Pi$. Consider the line graph with ordered set of locations $\{\pi(1),\cdots, \pi(m)\}$, where the customer visits the first $j$ locations with probability $\theta^\pi_j$, and visits none  with probability $1-\sum_{j\in[m]}\theta^\pi_j$. Using Lemma \ref{randline} on this line graph, we have that the sum of the expected revenues of $\log_2(\rho)$ (randomized) candidate placements is a constant factor of the optimal revenue. Since the candidate solutions constructed are independent of the specific permutation $\pi$, combining this result with the decomposition above yields Theorem \ref{randgen}.
} 

\begin{proof}{Proof of Theorem \ref{randgen}.}
Fix an arbitrary permutation $\pi\in \Pi$ and let $W^\pi(X)$ denote the expected total revenue from placement $X$ when customers visit a subset of locations $\{\pi(1), \pi(2),\cdots, \pi(m)\}$ in order, i.e., 
\[W^{\pi}(X)=\sum_{j\in[m]} \theta^{\pi}_jR(X^{\pi}(1:j)).\] 
Let $W^{\pi}(O)$ denote the expected total revenue of the optimal placement when customers visit locations in the order given by $\pi$. { Note that fixing $\pi$ reduces the problem to an instance of \placel\ where, for every $j\in[m]$, the customer visits the first $j$ locations with probability $\theta^\pi_j$, and visits none of the locations with probability $1-\sum_{j\in[m]}\theta^\pi_j$. From Lemma \ref{randline}, we have 
\[ \sum_{\ell=1}^{\log_2(\rho)}\ep[W^{\pi}(X^r_{2^{\ell-1}j_{\min}})]\geq \alpha (1-1/e)W^{\pi}(O),\]
here $W^{\pi}(O)$ is a lower bound on the optimal revenue for this instance of \placel. 
Now, the total revenue of placement $X$ can be written as,
\[\mathcal{R}(X)=\sum_{\pi\in\Pi} W^{\pi}(X).\]
Thus, for the randomized placement $Y$, we have,}
\begin{eqnarray*}
	(\log_2 \rho)\, \ep[\mathcal{R}(Y)]= (\log_2 \rho)\,\sum_{\pi\in \Pi}\ep[W^{\pi}(Y)]
	& \geq &\sum_{\pi\in \Pi}\sum_{\ell=1}^{\log_2 \rho}\ep[W^{\pi}(X^r_{2^{\ell-1}j_{\min}})],\\
	&\geq &\sum_{\pi\in \Pi}\alpha(1-1/e)\,W^{\pi}(O)=\alpha(1-1/e) \opt.		\end{eqnarray*}
\hfill\Halmos		\end{proof}
\subsection{Tightness of Analysis}\label{sec:tight}

In this section, we construct an instance of the \place\ problem to demonstrate the tightness of our algorithms. Recall that our algorithm depends on utilizing products from $S_k^*$, arranging them in a certain manner (randomized), and then selecting the best resulting placement for all $k \in [m]$. We will present an instance wherein, for any $\epsilon>0$, any placement of products in $S_k^*$ for any $k \in [m]$ yields a revenue that remains a factor $\Omega(\log^{1-\epsilon} m)$ away from the revenue of the optimal placement. This underscores the tightness of our algorithm's analysis in Theorem \ref{randgen}. We now describe the following instance.

\vspace{2mm}
\noindent
{\bf Instance:}
Let $\epsilon>0$. Consider a line with $m$ locations $\{1,2,\ldots,m\}$.  Denote $\theta_j$  the probability of visiting the subset of locations $\{1,2,\ldots,j\}$ and let $\theta_j= 1/ j ^{1+\frac{1}{1+\epsilon}}$. We consider an MNL choice model.
Consider a universe of products such that for each product  with a preference weight  $v$, we set the revenue of the product to $r = v^{-\frac{1}{1+\epsilon}}$. The universe is constituted with the following products:
\begin{itemize}
\item For each $j\in [m]$, there are $j$ different products with preference weights $v_j= \epsilon/ j $.
\item For each $j\in [m]$, there is a single product with preference weight $u_j= 1/( j \log^{1+\epsilon} j)$.
\end{itemize}
So in total we have a total of $m+ m(m+1)/2$ products in the universe.

The next theorem proves the tightness of analysis of our approximation algorithms for the \place\ problem. The proof of Theorem \ref{thm:tigh} is given in Appendix \ref{appx:hetero}.

\begin{theorem} \label{thm:tigh}
Consider the instance above with any $\epsilon>0$. For each $k \in [m]$, let $S^*_k$ be the optimal assortment of cardinality at most $k$. Let $X_k$ be the optimal placement using only products in $S^*_k$.
Let $ X^*$ be the optimal placement using any product from the universe. Then, for any $k \in [m]$, 
$$ \mathcal{R}(X^*)=\Omega(\log^{1-\epsilon}m)\mathcal{R}(X_k).$$
\end{theorem}

\subsection{Deterministic Algorithm for Markov Choice Model} \label{subsec:deran}
In practice, business owners may be hesitant to implement a randomized solution, even if one can guarantee that the solution is ``good" with high probability. In our context, the placement decision may not be likely to undergo frequent changes, increasing the importance of a deterministic solution. Our randomized algorithm can be run many times to generate multiple placement solutions and the best of these solutions will have revenue at least $\frac{O(\alpha)}{\log_2 \rho}\,\opt$ with high probability. We now show that given additional structure on the choice model, we can derandomize our algorithms to obtain a deterministic algorithm with the same approximation guarantee (for a general browsing distribution). In particular, we consider the family of (single-purchase) Markov choice models. This family includes several commonly used choice models, such as MNL and GAM as special cases, and also approximates MMNL and other choice models to a reasonable degree. Recall that, we say that an algorithm is deterministic if it does not use any randomness given a revenue oracle. Estimating the revenue requires randomization (due to random sampling) even for the ``deterministic" algorithm that we introduce below. 

Given assortment $S^*_k$, what is the best placement? This is a placement problem on the reduced ground set $S^*_k$. We show that, perhaps surprisingly, the problem is an instance of constrained monotone submodular maximization. Note that in Section \ref{sec:unif}, we showed a similar result for the setting where prices are homogeneous. For heterogeneous prices, it is well known that the revenue function is not monotone or submodular even for MNL choice model (special case of Markov). Our new result utilizes properties of optimal unconstrained assortments for Markov choice models (\cite{desir, SO}). 
\begin{lemma}\label{lem:redsubmod}
For every $k\in[m]$, there is a $(1-1/e)$-approximation algorithm for finding the optimal placement on the reduced ground set $S^*_k$ when $\phi$ is a Markov choice model.
\end{lemma}
We include the  proof of Lemma \ref{lem:redsubmod} in Appendix \ref{appx:deran}. Let $\hat{X}_k$ denote a $(1-1/e)$-approximate placement for $S^*_k$. Consider the placement solution,\[\hat{Y}=\argmax_{k\in[m]} \mathcal{R}(\hat{X}_k).\]
Observe that this is a deterministic solution. We establish the following approximation guarantee for $\hat{Y}$. The proof details of Theorem \ref{thm:morocco} are included in Appendix \ref{appx:deran}.
\begin{theorem} \label{thm:morocco}
The placement $\hat{Y}$ is a  $\frac{\Omega(1)}{\log_2 \rho}$-approximation algorithm for \place\ when $\phi$ is a Markov choice model.
\end{theorem}
Finding a deterministic $\frac{\Omega(\alpha)}{\log_2 m}$-approximation for general choice models is an interesting open problem.

{
\section{Numerical Experiments} \label{sec:numerics}

In this section, we present numerical experiments on synthetic data to evaluate the performance of our algorithm for the \place\ problem (Sections \ref{exp:line} and \ref{exp:grid}), {\color{black} as well as to assess the impact of product repetition (Section \ref{exp:repeat}).}

\subsection{Experimental Setup}

We start by introducing the universe of products \( N \), which consists of \( n \) products. The prices of the products, denoted as \( (r_i)_{i \in N} \), are generated from an exponential distribution with a parameter of 1. We assume that the choice model \( \phi \) follows a Multinomial Logit model where the preference weights of products \( (v_i)_{i \in N} \) are sampled from a uniform distribution over \([0,1]\).
We consider two different browsing distribution scenarios: Placement on a Line and Placement on a two-dimensional grid, as we describe next.

{\bf Placement on a Line:}
The line configuration corresponds to the case detailed in Section \ref{subsec:warmup} and illustrated in Figure \ref{fig:line}. Here, there are \( m \) locations arranged in a linear sequence. The browsing distribution is parameterized by \( m \) probabilities \( (\theta_j)_{1\leq j \leq m} \), where $\theta_j$ represents the likelihood of a customer visiting locations \( 1 \) (the topmost location) through \( j \). In our experiments, we generate \( \theta_j \) from a uniform distribution over \([0,1]\), followed by normalization to ensure that \( \sum_{j=1}^m \theta_j = 1 \). We perform the experiment with $n \in \{10,15,20\}$ and $m \in \{3,5,10,15,20\}$.

{\bf Placement on a Two-Dimensional Grid:}
We also consider a two-dimensional grid, where each location is represented by a point \((i,j)\), with \( i, j \in \{0, 1, \ldots, d\} \), for some non-negative integer $d$. 
The total number of locations is \( m = (d+1)^2 \). The browsing distribution is generated by performing a number \( \Gamma \) of random walks on the grid. Each customer randomly selects one of the $\Gamma$ walks to follow. The random walks obey the following rules: They begin at one of the four corners of the grid:  \((0,0)\), \((0,d)\), \((d,0)\), or \((d,d)\).
From any given location \((i,j)\), if \( i \notin \{0, d\} \) and \( j \notin \{0, d\} \), the walker can move to any of the four neighboring locations (up, down, left, right) with equal probability of $1/4$. If the walker is located on a boundary (e.g., \( i = 0 \)), movement is restricted to three possible directions, each with equal probability of $1/3$. The random walk terminates when the walker reaches a corner. An example of such a random walk is shown in Figure \ref{fig:random_walk}. 

This setup is inspired by the layout of a physical store, where the corners represent entrances and exits, and customers navigate through the store via random walks. {\color{black} We approximate the browsing distribution by generating 
$\Gamma$ random walks, each with equal probability.} 
For our numerical experiments, we choose $n \in \{10,15,20\}$ and  $ \Gamma =10$. We vary  $d \in \{1,2,3,4,5\}$ which means that $m \in \{4,9,16,25,36\}$.

\begin{figure}[h]
	\centering
	\includegraphics[width=0.5\linewidth]{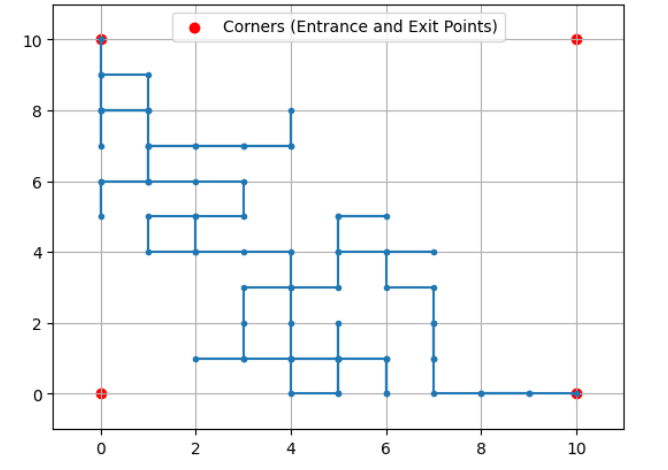}
	\caption{Example of a random walk on the two-dimensional grid.}
	\label{fig:random_walk}
\end{figure}

\subsection{Results}

For each value of $n$ and $m$, we sample 100 instances of the \place\ problem according to the distributions described above. For each instance, we solve the \place\ problem optimally by developing a mixed integer linear programming (MILP) formulation. The details of our formulation are provided in Appendix \ref{appendix:IP}. Note that the MILP formulation uses big-M constraints and can accommodate general browsing distributions. The formulation can be solved efficiently for relatively small instances but it takes a significant amount of time to solve larger ones. {\color{black} Note that because the browsing distribution is approximated via a finite set of sampled random walks, the MILP we solve is in fact an instance of a sample average approximation (SAA)–based integer program, i.e., the browsing distribution is represented by the empirical distribution over the sampled paths.}
We compare the optimal solution with the solution generated by our algorithms. For the line, we use the deterministic algorithm described in Section \ref{subsec:warmup}. For the two-dimensional grid, we use the general randomized algorithm described in Section \ref{subsec:general}. {\color{black} We use Gurobi 10.0.2 as the integer programming solver, and the numerical experiments were conducted in Python on a MacBook Pro equipped with a 2.3-GHz 8-core Intel Core i9 processor and 32 GB RAM.}

\vspace{2mm}
\noindent
\subsubsection{Placement on a Line.}\label{exp:line} The results are reported in Table \ref{tab:results_line}. The first and second columns in Table \ref{tab:results_line} present the values of $m$ and $n$. The third and fourth columns provide the ratio of the objective value obtained by our algorithm to the optimal objective value. Out of the 100 instances, we report the minimum ratio and the mean ratio. The fifth column shows the running time of our algorithm in seconds, while the sixth column reports the running time of the MILP formulation, also in seconds.

We observe that our algorithm performs quite well, achieving an average ratio of 0.97 or higher for all values of $m$ and $n$. Furthermore, the worst-case ratio is also very good, with a minimum value of 0.84. Note that, in theory, our algorithm achieves a worst-case ratio of $1/\log(m+1)$ for a line (as shown in Lemma \ref{line}), but the performance seems much better for the randomly  generated instances in our experiments. Additionally, our algorithm is extremely fast, with a running time of less than one second. In contrast, the running time of the MILP approach increases exponentially with $m$ and $n$, becoming significantly higher as the dimensions increase. For instance, it takes an average of 5 minutes to solve cases where $m = n = 20$, and for larger values, it exceeds practical time limits. This is due to the big-$M$ constraints and the large number of variables and constraints in the formulation.

\begin{table}[t]

	\centering
	\begin{tabular}{|c|c|cc|cc|}
		\hline
		\textbf{$n$} & \textbf{$m$} & \multicolumn{2}{|c|}{\textbf{Ratio}} & \multicolumn{2}{|c|}{\textbf{Time(s)}} \\ 
		&             & \textbf{Min} & \textbf{Mean}       &                        \textbf{ALG} & \textbf{OPT}                        \\ \hline
		\multirow{3}{*}{10} & 3  & 0.861 & 0.981 & 0.03 & 0.05 \\ \cline{2-6}
		& 5  & 0.842 & 0.983 & 0.06 & 0.18 \\ \cline{2-6}
		& 10 & 0.905 & 0.987 & 0.12 & 0.87 \\ \hline
		\multirow{4}{*}{15} & 3  & 0.866 & 0.978 & 0.04 & 0.10 \\ \cline{2-6}
		& 5  & 0.866 & 0.985 & 0.07 & 0.39 \\ \cline{2-6}
		& 10 & 0.895 & 0.986 & 0.15 & 6.45 \\ \cline{2-6}
		& 15 & 0.937 & 0.989 & 0.24 & 16.88 \\ \hline
		\multirow{5}{*}{20} & 3  & 0.854 & 0.975 & 0.04 & 0.15 \\ \cline{2-6}
		& 5  & 0.851 & 0.977 & 0.07 & 0.64 \\ \cline{2-6}
		& 10 & 0.909 & 0.986 & 0.17 & 27.72 \\ \cline{2-6}
		& 15 & 0.926 & 0.990 & 0.24 & 131.81 \\ \cline{2-6}
		& 20 & 0.942 & 0.992 & 0.36 & 305.86 \\ \hline
	\end{tabular}
	\caption{Results of experiments on Line graphs. For each $(n,m)$ pair, we report the minimum and mean performance of ALG relative to OPT and the average running times.}
	\label{tab:results_line}
\end{table}

\begin{table}[t]
    \centering
    {\color{black}
    \begin{tabular}{|c|c|c|cc|cc|c||cc|cc|}
        \hline
        \textbf{$n$} & \textbf{$d$} & \textbf{$m$}
        & \multicolumn{2}{c|}{\textbf{Ratio}}
        & \multicolumn{2}{c|}{\textbf{Time(s)}}
        & \shortstack{\textbf{\% Instances}}
        & \multicolumn{2}{c|}{\textbf{Occurrences}}
        & \multicolumn{2}{c|}{\textbf{Repet. Gain (\%)}} \\
        & & & \textbf{Min} & \textbf{Mean} & \textbf{ALG} & \textbf{OPT}
        & \textbf{$<$1H}  & \textbf{ALG} & \textbf{OPT} & \textbf{Mean} & \textbf{Max} \\ \hline
        \multirow{5}{*}{10}
        & 1 & 4  & 0.740 & 0.960 & 0.12 & 0.12 & 100 & 2.01 & 2.07 & 21.1 & 77.8 \\ \cline{2-12}
        & 2 & 9  & 0.858 & 0.958 & 0.14 & 1.40 & 100 & 3.93 & 2.94 & 31.5 & 63.1 \\ \cline{2-12}
        & 3 & 16 & 0.839 & 0.962 & 0.15 & 2.42 & 100 & 6.42 & 4.10 & 36.2 & 71.3 \\ \cline{2-12}
        & 4 & 25 & 0.882 & 0.968 & 0.16 & 2.02 & 100 & 9.86 & 5.42 & 42.0 & 129.7 \\ \cline{2-12}
        & 5 & 36 & 0.877 & 0.968 & 0.16 & 28.47 & 100 & 15.7 & 6.95 & 43.6 & 81.7 \\ \hline
        \multirow{5}{*}{15}
        & 1 & 4  & 0.841 & 0.960 & 0.19 & 0.25 & 100 & 1.88 & 1.96 & 18.1 & 39.8 \\ \cline{2-12}
        & 2 & 9  & 0.868 & 0.958 & 0.21 & 11.79 & 100 & 3.47 & 2.72 & 27.6 & 74.5 \\ \cline{2-12}
        & 3 & 16 & 0.897 & 0.958 & 0.22 & 18.08 & 100 & 5.55 & 3.84 & 33.2 & 65.3 \\ \cline{2-12}
        & 4 & 25 & 0.844 & 0.965 & 0.26 & 37.63 & 100 & 8.64 & 5.39 & 35.2 & 81.4 \\ \cline{2-12}
        & 5 & 36 & 0.855 & 0.964 & 0.26 & 27.05 & 100 & 13.3 & 6.80 & 37.6 & 79.9 \\ \hline
        \multirow{5}{*}{20}
        & 1 & 4  & 0.851 & 0.960 & 0.27 & 0.35 & 100 & 1.67 & 1.96 & 14.1 & 50.8 \\ \cline{2-12}
        & 2 & 9  & 0.877 & 0.953 & 0.30 & 24.27 & 100 & 3.29 & 2.61 & 24.4 & 49.1 \\ \cline{2-12}
        & 3 & 16 & 0.861 & 0.959 & 0.32 & 165.45 & 94  & 5.70 & 3.72 & 30.5 & 51.6 \\ \cline{2-12}
        & 4 & 25 & 0.897 & 0.967 & 0.33 & 134.46 & 89  & 8.41 & 4.91 & 30.7 & 49.0 \\ \cline{2-12}
        & 5 & 36 & 0.882 & 0.967 & 0.39 & 225.64 & 89  & 11.9 & 6.45 & 35.5 & 56.5 \\ \hline
    \end{tabular}
    \caption{Results of experiments on two-dimensional grids. For each $(n,d)$ pair, we report the minimum and mean performance of ALG relative to OPT, average running times, and the percentage of instances solved by the MILP within one hour. The last few columns show the extent of product repetition and the revenue gains from allowing repetition (see Section \ref{exp:repeat}).}
    \label{tab:results_random_walk}}
\end{table}

\vspace{2mm}
\noindent
\subsubsection{Placement on a Two-Dimensional Grid.}\label{exp:grid} The results  are reported in Table \ref{tab:results_random_walk}. The first, second, and third columns in Table \ref{tab:results_random_walk} present the values of $n$, $d$, and $m$. 
For each pair $(n, d)$, we run 100 instances of the placement problem. We observe that the running time of the MILP used to compute the optimal solution of the placement problem can be extremely high; therefore, we limit the running time to 1 hour. In the eight column of the table (\% Instances $<$1H), we report the percentage of instances that are solved to optimality within the 1-hour limit. For each instance where the optimal solution is computed within the time limit, we generate a solution using the algorithm described in Section \ref{subsec:general}, and then calculate the minimum and average ratios of the objective values over all such instances. Note that we execute our randomized algorithm once for each instance. The fourth and fifth columns of the table (Ratio Min and Ratio Mean) report these ratios: the minimum ratio and the mean ratio of the objective value obtained by our algorithm to the optimal objective value, evaluated over the instances solved to optimality. The sixth column (Time(s) \alg) shows the average running time of our algorithm in seconds. The seventh column (Time(s) \opt) reports the average running time of the MILP approach, also in seconds, with the average calculated only over the instances that are solved to optimality.

Our algorithm achieves an average ratio of 0.96 across all values of $n$ and $d$. In the worst case, the minimum ratio is 0.74, achieved for the smallest instance with $n=10$ and $d=1$. In fact, as we increase the problem size, it seems that the performance of the algorithm improves. Note that the ratios are slightly smaller than those for Placement on a Line, but they are still very high. In theory, our algorithm achieves a worst-case ratio of $(1 - 1/e)/\log(m+1)$ for the grid (as shown in Theorem \ref{randgen}), but we observe significantly better empirical performance. Furthermore, our algorithm is extremely fast, with a running time of less than one second. In contrast, the running time of the MILP approach increases exponentially with $m$ and $n$, becoming significantly higher as the dimensions increase. In fact, obtaining the optimal solution becomes extremely challenging as the dimensions grow. For example, while we are able to solve all instances for $n=10$ and $n=15$ within the 1-hour time limit, solving instances for $n=20$ becomes much harder. For instance, when $n=20$ and $d=5$, 11 instances did not converge to optimality within the 1 hour time limit. The MILP formulation, with its large number of variables and constraints and big-$M$ constraints, becomes impractical. This highlights the necessity of our algorithm, which can quickly compute good solutions.

{\color{black}
\vspace{2mm}
\noindent
\subsubsection{Value of Repetition in the Grid.}\label{exp:repeat} Recall that in our \place\ problem we allow products to be repeated in multiple locations. In the \placel\ setting, such repetition does not provide any benefit (see Remark~\ref{pline}). In the worst case, not allowing product repetition can lead to significantly lower revenue (see Example \ref{ex1}). To quantify the significance of repetition beyond the worst case, we numerically evaluate its effect in the grid setting.  

We first calculate for each product the number of times it is duplicated across locations (we refer to this as its number of occurrences), and we average this quantity over products with at least one occurrence. We compute this measure both for the optimal solution (OPT) and for the solution returned by our algorithm (ALG) and report these averages in the ninth and tenth columns of Table~\ref{tab:results_random_walk} (Occurrences \alg\ and Occurrences \opt).
The results show that both OPT and ALG exhibit duplication, with an average number of occurrences equal to  4.12 for OPT and 6.77 for ALG across all our instances.  As expected, when the number of locations increases while the number of products is fixed, we observe more duplication; conversely, increasing the number of products for a fixed number of locations reduces the extent of duplication. Furthermore, OPT systematically uses fewer duplicates than ALG. This is natural: OPT leverages global structural considerations to minimize unnecessary repetition, while ALG, by construction, samples uniformly from a candidate subset at each location, which tends to generate more duplicates.

To evaluate the value of repetition, we compare the optimal revenue from the \place\ problem (i.e., the optimal objective value of  \place\ obtained through the MILP in Appendix~E) with the optimal revenue of the same problem when repetition is disallowed. The latter is formulated by simply adding a constraint to the MILP in Appendix~E requiring that each product appear in at most one location; we denote its optimal objective value by $\text{OPT}_{\text{norep}}$. By construction, $\text{OPT} \geq \text{OPT}_{\text{norep}}$. We define the \emph{Repetition Gain} as
$
\frac{\text{OPT} - \text{OPT}_{\text{norep}}}{\text{OPT}_{\text{norep}}},
$
and report this quantity in the last two columns of Table~\ref{tab:results_random_walk} (Repet. Gain (\%)), both as an average over all 100 instances (Mean) and as the maximum observed value (Max).  
The results indicate that repetition provides a significant benefit. On average, the gain ranges from roughly 14\% to 43\%, depending on the values of $n$ and $m$, and in some cases the maximum gain can be much higher. For example, when $n=10$ and $d=4$, the repetition gain reaches 129\% in average, representing the highest value observed in our experiments. These findings confirm that, unlike in the line setting where duplicates play no role, repetition in the grid setting can substantially improve revenues by amplifying the placement of revenue generating products across multiple visible locations.

\vspace{2mm}
\noindent
\emph{Structure of Solutions:} To further understand the structure of the solutions, we look at how product prices influence the locations where products are placed and the number of occurrences. We report our results in Figures~\ref{fig:occ-scores} and~\ref{fig:vis-scores} in Appendix \ref{appx:exp}, and summarize our main findings here. 
Our experiments reveal several consistent patterns in the structure of \opt\ and \alg. Both  exhibit a clear ``revenue-ordered" structure: the highest-revenue products are duplicated most often and are preferentially assigned to the most visible locations, with duplication frequency and visibility decreasing as product revenue decreases. While this ordering is stable across different values of $n$ and $m$, ALG tends to employ more duplication than OPT, reflecting its randomized design. Overall, these findings confirm that both methods exploit the MNL model’s tendency to favor revenue-ordered assortments, with ALG trading off optimality for simplicity by duplicating high-revenue products more aggressively.

}

}


\section{Conclusion}
We study the problem of finding an optimal placement of substitutable products by introducing a model where customer choice is a two stage process. First, the customer browses the store and visits a subset of display locations according to a known browsing distribution. Then, the customer chooses from the assortment of products located at the visited locations according to a discrete choice model. We consider the problem for a general browsing distribution and general discrete choice model, a setting for which no approximation results were known in the literature (to the best of our knowledge). We leverage the existing results for cardinality constrained assortment optimization to obtain provably good placement of products in this general setting. In particular, we develop an approximation algorithm that uses an $\alpha$-approximation algorithm for cardinality constrained assortment optimization to generate a $\frac{\Theta(\alpha)}{\log_2 \rho}$-approximation for the placement problem for any browsing distribution. Our algorithm is randomized and we show that it can be derandomized for the Markov choice model. 
When products are homogeneously priced, we show that the problem is an instance of monotone submodular function maximization subject to a partition matroid constraint. This gives a $(1-1/e)$-approximation for our setting in the special case of homogeneous prices. We show that unless P=NP, no polynomial time algorithm can beat this guarantee.

Improving the $\frac{1}{\log_2 \rho}$ factor in our guarantee or showing that no efficient algorithm has a better guarantee remains a challenging open problem. Our model assumes that the browsing distribution does not depend on the placement solution. This may not be true in general, especially in settings such as cultural institutions~(\cite{cultural}), and finding an approximation algorithm without this assumption remains a challenging open problem.

{
\bibliographystyle{informs2014.bst}
\bibliography{assortbib}}

\newpage
\begin{APPENDICES}

{
	\section{Missing Proofs from Section \ref{sec:prelim}}\label{appx:linehard}

}	

\subsection{Proof of Lemma \ref{est}} \label{apx:est}
    
	\begin{proof}{Proof of Lemma \ref{est}.}
		Given a placement $X$, consider $T$ independent samples obtained from the sampling oracle. Let $Y_t$ denote the expected revenue of sample $t\in [T]$. Let \[R^*=\max_{S\subseteq N,\, |S|\leq m} R(S).\] 
		Note that $Y_t\leq R^*$, as $R^*$ is the (expected) revenue of a revenue optimal assortment of size at most $m$. Let $\hat{\mathcal{R}}=\frac{1}{T}\sum_{t\in [T]} Y_t$. Notice that the random variables $\{\frac{Y_t}{R^*}\}_{t\in [T]}$ have range $[0,1]$. Applying Hoeffding's bound (Lemma \ref{chernoff}) on the random variables $\{\frac{Y_t}{R^*}\}_{t\in [T]}$ with $t$ equals $\eta T$, we get
		$$
		\mathbb{P}\left(\Bigg|\sum_{t \in [T]} \frac{Y_t}{R^*} -  \sum_{t \in [T]}  \frac{\mathbb{E}[Y_t]}{R^*}\Bigg|\leq \eta T\right)\geq 1-2e^{-2T\eta^2}\qquad \forall \eta\geq 0,
		$$
		which is equivalent to
		\begin{equation}\label{chern1}
			\mathbb{P}\left(|\hat{\mathcal{R}}-\mathcal{R}(X)|\leq \eta R^*\right)\geq 1-2e^{-2T\eta^2}\qquad \forall \eta\geq 0.
		\end{equation}
		Substituting $T=\frac{m^2\log \delta^{-1}}{2\epsilon^2}$ and $\eta=\epsilon/m$ in \eqref{chern1} and using the fact that $R^*\leq m\opt$, we have
		\[   \mathbb{P}\left(|\hat{\mathcal{R}}-\mathcal{R}(X)|\leq \epsilon\, \opt\right)\geq 1-2\delta,\]
		as desired. 
		\hfill\Halmos\end{proof}
	\begin{lemma}\label{chernoff}{(\cite{hoeffding1994probability})}
		Let $X_1,X_2,\cdots, X_T$ be independent random variables with range $[0,1]$ and let $\mu$ denote the expected value of the sum $\sum_{i=1}^T X_i$. Then, for all $t>0$, 
		\[\mathbb{P}\left(\bigg|\sum_{i=1}^T X_i-\mu \bigg|\leq t\right)\geq 1-2e^{-2t^2/T}.\]
	\end{lemma}

 \subsection{Proof of Lemma \ref{err}}    \label{apx:err}
	\begin{proof}{Proof of Lemma \ref{err}.}
		The following chain of inequalities proves the claim,
		\[\mathcal{R}(\hat{X}^*)+\epsilon\, \opt\geq \hat{\mathcal{R}}(\hat{X}^*)\geq \hat{\mathcal{R}}(X^*)\geq \mathcal{R}(X^*)-\epsilon\, \opt.\]
		The first and last inequalities follow from Corollary \ref{coro} where we show  that $\hat{\mathcal{R}}(\cdot)$ is a good estimate of $\mathcal{R}(\cdot)$ for every placement in the set $X_1,X_2,\cdots, X_t$. The second inequality follows by definition of $\hat{X}^*$. 
		
		\hfill\Halmos		\end{proof}

	\section{Missing Proofs from Section \ref{sec:homo}}\label{appx:homo}

    \subsection{Proof of Theorem \ref{uniform}} \label{apx:uniform}
	\begin{proof}{Proof of Theorem \ref{uniform}.}
		Recall $N$ is the set of products and $G$ is the set of locations.
		Consider the expanded ground set $N\times G=\{(i,j)\mid i\in N, j\in G\}$. A placement $X$ corresponds to set $\left\{(X(j),j)\mid j\in G \right\}\subseteq N\times G$. Given a set $U\subseteq N\times G$, for each $j \in G$, we define the following set of products
		\[U(j)=\{  i\in N     \mid (i,j)\in U \}.\] 
		Consider the matroid with family of independent sets $\mathcal{F}=\{U\mid\, |U(j)|\leq 1\,\, \forall j\in G\}$. Each placement $X$ corresponds to a unique set in $\mathcal{F}$. Similarly, given a set $U\in \mathcal{F}$, we have a unique placement solution where location $j\in G$ has product $i$ if $(i,j)\in U$ for some $i$, and $j$ has no product otherwise.
		Finally, given an arbitrary set $U\subseteq N\times G$, let
		\begin{equation}\label{decompt}
			\mathcal{R}(U)=\sum_{L \in 2^G}  \mathbb{P}_B(L) \,\, R\left(\cup_{j \in L} U(j)\right).
		\end{equation}
		
		The placement problem is equivalent to the following matroid constrained subset selection problem on ground set $N\times G$, 
		\[\max_{U\in \mathcal{F}} \mathcal{R}(U).\]
		It remains to show that $\mathcal{R}(\cdot)$ is monotone and submodular. Fix a subset of locations $L \in 2^G$ and  define the function $f_L(U)= R(\cup_{j \in L} U(j)) $ for any $U \subseteq N \times G$. We know from equation \eqref{decompt} that $\mathcal{R}(.)$ is a linear combination of $f_L$ for $L \in 2^G$. So it suffices to show that $f_L$ is monotone submodular.
		
		Let $U \subseteq  N \times G$ and $ V \subseteq N \times G$. If $U \subseteq V$, then for any $L \in 2^G$, we have $\cup_{j \in L} U(j) \subseteq \cup_{j \in L} V(j)$. Therefore, by monotonicity of $R(.)$  from Lemma \ref{submod}, we get 
		$$f_L(U) =  R\left(\cup_{j \in L} U(j)\right) \leq  R\left(\cup_{j \in L} V(j)\right) =f_L(V),$$ which implies that  $ f_L$ is monotone. Now let us show submodularity. Consider, $(t,s) \in N \times G    $. If $s \notin L$, then we clearly have
		$f_L(U \cup (t,s))= f_L(U )$ and $f_L(V \cup (t,s))= f_L(V )$. Now, if $s \in L$,  by submodularity of $R(.)$ from Lemma \ref{submod}, we get
		\begin{align*}
			f_L(U \cup (t,s))- f_L(U ) & =  R \left(\cup_{j \in L} U(j) \cup t    \right) - R\left(\cup_{j \in L} U(j)\right) \\ &\geq R \left(\cup_{j \in L} V(j) \cup t    \right) - R\left(\cup_{j \in L} V(j)\right) \\
			&=  f_L(V \cup (t,s))- f_L(V ) ,   
		\end{align*}
		which implies that $f_L$ is submodular.

		Therefore, we can obtain  $(1-1/e)$ approximation for the \place\ problem using submodular maximization subject to a matroid constraint~(\cite{calinescu2011maximizing, badanidiyuru2014fast, filmus2012tight}). For submodular function maximization, it is well known (for example, see \cite{calinescu2011maximizing}) that when we use estimates $\hat{f}_L(U)\in [f_L(U)-\epsilon \opt, f_L(U)+\epsilon \opt]$ of the function value instead of exact value then the final output is $(1-1/e)-O(n\epsilon)$ approximate. Thus, replacing the exact revenue oracle with sampling based estimates discussed in Section \ref{sec:sample} gives us a reduced guarantee of $(1-1/e)-O(n\epsilon)$. Let $\text{poly}(n,m)$ denote a polynomial in $n$ and $m$. From Lemma \ref{est}, using $O(\text{poly}(n,m)\log \delta^{-1})$ samples from the sampling oracle gives us estimates with error parameter $\epsilon=o(1)/n$, with probability at least $1-\delta$, resulting in approximation guarantee of $(1-1/e)-o(1)$. 
		\hfill\Halmos		\end{proof}

	    \subsection{Proof of Corollary \ref{coro:one}} \label{apx:coro}
	\begin{proof}{Proof of Corollary \ref{coro:one}.}
		We partition $N$ into $\Delta=\floor{\log \frac{r_{\max}}{r_{\min}}}+1$ disjoint sets $\{S_1,S_2,\cdots, S_{\Delta}\}$ such that,
		\[S_j=\{i\mid 2^{j-1} r_{\min}\leq r_i < 2^{j}r_{\min}\}\qquad \forall j\in[\Delta].\]
		For $j\in[\Delta]$, let $\mathcal{X}_j$ denote the set of all placements using only the products in $S_j$ and let $\bar{r}_j=2^{j}r_{\min}$. Since every product in $S_j$ has a price in the range $[0.5\bar{r}_j,\bar{r}_j)$, from Theorem \ref{uniform} 
		we have an algorithm for finding a placement solution $Y_j\in \mathcal{X}_j$ such that,
		\begin{equation}\label{given}
			\mathcal{R}(Y_j)\geq 0.5(1-1/e)\max_{X\in \mathcal{X}_j} \mathcal{R}(X),
		\end{equation}
		here the additional factor of $0.5$ comes from the fact that $r_e\geq 0.5 \bar{r}_j\,\, \forall e\in S_j$.
		Now, let $X^* \in \mathcal{X}$ denote the optimal placement of products in $N$. We will use the standard fact that the revenue function $R(\cdot)$ is subadditive for a monotone and substitutable choice model, i.e.,
		\[R(A\cup B)=\sum_{i\in A\cup B}r_i\phi(i,A\cup B)\leq\sum_{i\in A} r_i \phi(i,A)+\sum_{i\in B}r_i \phi(i,B)= R(A)+R(B)\qquad \forall A,B\subseteq N.\]
		Using this in combination with \eqref{given}, we have
		\begin{eqnarray*}
			\mathcal{R}(X^*)&=&\sum_{L\in 2^G} \mathbb{P}_{B}(L)\,\, R(X^*(L)),\\
			&\leq & \sum_{L\in 2^G} \mathbb{P}_{B}(L)\,\, \left[\sum_{j\in[\Delta]} R(S_j\cap X^*(L))\right],\\
			&= &\sum_{j\in[\Delta]}\left[\sum_{L\in 2^G} \mathbb{P}_{B}(L)\,\, R(S_j\cap X^*(L))\right],\\
			&\leq &\sum_{j\in[\Delta]} \max_{X_j\in \mathcal{X}_j} \mathcal{R}(X_j),\\
			&\leq &\frac{1}{0.5(1-1/e)}\sum_{j\in[\Delta]} \mathcal{R}(Y_j) ,\\
			&\leq &\frac{\Delta}{0.5(1-1/e)} \cdot \max_{j\in [\Delta]}\mathcal{R}(Y_j).
		\end{eqnarray*}
		Therefore, the best placement out of $Y_j$ for $j \in [\Delta]$ gives  $\frac{O(1)}{\log \frac{r_{\max}}{r_{\min}}}$-approximation for  \place .
		\hfill\Halmos		\end{proof}

	\subsection{ Proof of Theorem \ref{thm:hardness}}\label{appx:hard}

	
	\begin{proof}{Proof of Theorem \ref{thm:hardness}.}
		
		We use a reduction from the maximum coverage problem in a similar fashion that has been used to show the inapproximability result for the {\em customized assortment problem} in \cite{el2023joint}. The problem in \cite{el2023joint} is slightly different but the proof can be adjusted to the setting of our theorem. We are providing the details below for completeness. We start by presenting the maximum coverage problem.

		\noindent
		{\em{Maximum coverage problem.}} Given elements $\{1,2,\ldots, q\}$ and sets $\{ S_1,S_2, \ldots,  S_n \}$ with \mbox{$S_j \subseteq \{1,2,\ldots,q\}$} for each $j=1,\ldots,n$, we say that set $S_j$ covers element $i$ if $i \in S_j$. For a given integer $K$,
		the goal of the maximum coverage problem is to find at most $K$ sets such that the total number of covered elements is maximized. This problem is NP-hard to approximate within a factor better than $(1- \frac{1}{e})$ unless $\text{P} = \text{NP}$ (\cite{feige1998threshold}).

		Fix $\epsilon>0$. Consider an instance of the maximum coverage problem with sets $\{S_1,\ldots,S_n\}$, elements $\{1,\ldots,q\}$ and an integer $K$. 
		We construct an instance of cardinality constrained assortment optimization  as follows. The choice model is given by MMNL. There are $n$ products and $q$ customer types. We use $N$ to denote the set of products and $[q]$ to denote the set of customer types. The products correspond to the sets $\{S_1,\ldots,S_n\}$ and the customers types correspond to the elements $\{1,\ldots,q\}$. All the products have the same price $r_i=1$ for all $i \in N$.  Denote $v_{ij}$ the preference weight of product $i$ for customer type $j$ under the MMNL model. Let $M=\frac{1}{\epsilon}-1$. For all $i\in N$ and $j \in [q]$, let
		$$
		v_{ij} = \left\{
		\begin{array}{ll}
			M  & \mbox{if }  j \in S_i \\
			0  & \mbox{otherwise.}
		\end{array}
		\right.
		$$
		The arrival probabilities of the customer types are given by $\theta_j=1/q$ for all $j \in [q]$. Note that the cardinality constrained assortment optimization problem under MMNL with uniformly priced products is given by 
		\begin{equation*}\label{mmnl} \tag{\sf MMNL}
			\begin{aligned}
				\max_{S \subseteq N, \;\vert S \vert \leq K \; }  \; \;   \frac{1}{q}   \sum_{j\in [q]}    \frac{  \sum_{i \in S}   v_{ij} }{1+\sum_{i \in S}  v_{i j} }.
			\end{aligned}
		\end{equation*}

		Suppose there exists a maximum coverage solution with an objective value $z$, i.e., there are $z$ covered elements. We construct a solution for \ref{mmnl} using exactly the $K$ products corresponding to the $K$ sets in this  maximum coverage solution. There are $z$ customer types that are covered by these sets. Hence, for each customer type $j$ among these $z$ customer types, there exists a product~$i$ such that $v_{ij}=M$. Therefore, we get at least an expected revenue of $\frac{M}{1+M}$ from each of these customer types.
		Let $R$ be the objective value of our solution for \ref{mmnl}. We have
		$$ R \geq \frac{1}{q} \cdot  \frac{M}{1+M} \cdot z = \frac{1}{q} \cdot (1-\epsilon) z.$$
		
		Now, let us consider a solution of \ref{mmnl}. Without loss of generality, the solution has $K$ products and let  $R$ be its objective value. We construct a feasible solution for the maximum coverage problem by choosing exactly the   sets corresponding to the $K$ products in the solution of \ref{mmnl}. Let  $z$ be the resulting total number of covered elements, and for each $j \in [q]$, let $\gamma_j$ be the number of sets that cover element $j$.
		If an element $j$ is not covered, then $\gamma_j =0$ and the expected revenue of the customer type $j$ is 0. Otherwise, if $\gamma_j \geq 1$, then the expected revenue from   customer type $j$  is given by $  \frac{\gamma_j M}{1+\gamma_j M}  \leq 1$. 
		Therefore, the objective value of our solution for \ref{mmnl} is given by
		$$R = \frac{1}{q} \sum_{j\in [q]} \frac{\gamma_j M}{1+\gamma_j M}    \leq \frac{1}{q} \cdot z .$$

		Consequently, a $\beta$-approximation for the \place\ problem implies a $\beta (1-\epsilon)$-approximation for the maximum coverage problem.
		We know that unless $\text{P} = \text{NP}$, it is NP-hard to approximate the maximum coverage problem with a factor better than $(1-\frac{1}{e})$ (\cite{feige1998threshold}). Thus, it is NP-hard to approximate the \place\ problem within a factor better than $(1- \frac{1}{e})/(1-\epsilon) \geq (1- \frac{1}{e})(1+\epsilon)  $ for any $\epsilon>0$. By rescaling $\epsilon$, we conclude it is NP-hard to approximate the \place\ problem within a factor better than $(1- \frac{1}{e}+ \epsilon)  $ for any $\epsilon>0$.

		\hfill\Halmos	\end{proof}

\section{Missing Proofs from Section \ref{sec:hetero}} \label{appx:hetero}

   \subsection{Proof of Lemma \ref{lem:hardness} } \label{apx:lem:hardness}
\begin{proof}{Proof of Lemma \ref{lem:hardness}.}
\cite{gallego2020approximation} formulate a product framing problem for MNL choice model and show that it is weakly NP-hard (see Theorem 1 in \cite{gallego2020approximation}). To show that \placel\ is NP-hard, it suffices to reduce product framing to \placel.

To distinguish the two settings, we refer to locations in the framing problem as \emph{pages}. A page may carry multiple products. Consider an arbitrary instance of framing with $m$ pages on a line, each with up to $p$ products, probability $\theta_j$ that the customer sees all products on the first $j$ pages (for $j\in [m]$), and an MNL choice model $\phi$. In the corresponding instance of \placel\ we have $m\, p$ locations and choice model $\phi$. The browsing distribution has support only on sets $[j\, p]$ for $j\in [m]$. In particular, the probability that customer visits the first $jp$ locations equals $\theta_j$ (for $j\in [m]$). This instance of \placel\ is not immediately equivalent to the instance of the framing problem because in framing: $(i)$ each product is placed on at most one page and $(ii)$ some pages may be left empty. 

Let $X_f$ denote an optimal solution for the framing problem with revenue $R_f$. Similarly, let $X_l$ denote an optimal solution for \placel\ with revenue $R_l$. It suffices to show that, $(a)$ one can (in polynomial time) convert $X_f$ into a feasible solution for \placel\ with revenue at least $R_f$, and $(b)$ the solution $X_l$ can be converted (in polynomial time) into a feasible solution for framing with revenue at least $R_l$. 

First, we convert $X_f$ to a solution for \place\ by placing the products on page $j$ at locations $ \{j\,p+1, j\, p+2,\cdots, (j+1)\,p\}$ (the exact order of placement does not affect revenue). This is a feasible solution for \place\ since each location has at most one product. 

We can convert $X_l$ into a solution for framing by placing on page $\ell$ the (up to) $p$ products in locations $\{\ell\,p+1,\ell\,p+2,\cdots, (\ell+1)p\}$. While this solution has revenue $R_l$, it may not be feasible for the framing problem because some products may be repeated. Therefore, we need to convert $X_l$ into a placement solution $X_{l^*}$ where no product is repeated and $R_{l^*}\geq R_{l}$.
We obtain $X_{l^*}$ by starting with $X_l$ and for each product that is repeated in $X_l$, we remove the product from all locations except the first location that the product is shown. We claim that $R_{l^*}=R_{l}$, i.e., removing product copies does not affect the revenue. To see this, suppose that a product is placed at two locations $j_1$ and $j_2$ where $j_2$ is below location $j_1$ on the line. Since every customer who visits $j_2$ also visits $j_1$, removing the product from $j_2$ does not change the revenue. 

\hfill\Halmos	\end{proof}



\subsection{Proof of Lemma \ref{single}}

\begin{proof}{Proof of Lemma \ref{single}.} 
Fix $k \in [m]$. First, let us show that the optimal placement of $S^*_k$ is at most $ \frac{O(1)}{k} \cdot \opt$ in the worst case. 
Consider the following instance on a line and choice given by MNL model. We use $v_i$ to denote the preference weight of product $i$ under the MNL model. We have $m$ locations in the line: $\{1,2,\ldots,m\}$ with $1$ being the topmost. Recall $\theta_j$ is the probability of visiting the subset of locations $(1:j)$. Let $\theta_1=1$ and $\theta_j=0$ for $j\in \{2,\ldots,m\}$. 
Let $N$ be a universe of $k+1$ distinct products with $N=U \cup \{q\}$. There are $k$ products in $U$ where for all $i \in U$, $r_i=k$  and $v_i=1/k$. Additionally, $r_q=k/2$ and $v_q=1$. Under MNL, we can compute optimally cardinality-constrained assortments.  So,  $S^*_k$ here denotes the optimal assortment of size $k$ under MNL.
First, let us show that $S^*_k=U$. We have $R(U)=k/2$. Moreover, let $p<k$ and let $S \subseteq U $ such that $S$ has $p$ products. Since all products in $U$ are similar, it does not matter which $p$ products we pick. We have
$R(S \cup \{q\})= \frac{p+k/2}{1+p/k +1}$, which, for $p<k$, is increasing in $p$. This implies that  $R(S \cup \{q\})< k/2$. We can also verify that $R(S )< k/2$ for any $p$. Therefore, we have that $S_k^*=U$. Now consider $X_1$ a placement of $U$ on the line. Note that all locations have $0$ visiting  probability except the first one. Since, all products in $U$ are similar, it doesn't matter which product we place in location $1$, we get the following revenue
$$\mathcal{R}(X_1) =  \frac{1}{1+1/k}= \frac{k}{k+1}.$$
On the other hand, denote $X_2$ a placement of product $q$ in the first location. The rest of locations are empty. We have
$$\mathcal{R}(X_2) =  R(\{q\})=  \frac{k/2}{1+1}= \frac{k}{4}.$$
Therefore, for any placement $X_1$, we have $\mathcal{R}(X_1) = \frac{O(1)}{k}  \cdot \mathcal{R}(X_2)    \leq \frac{O(1)}{k} \cdot \opt $.


Second, let us show that the optimal placement of $S^*_k$ is at most $O(1) \cdot \frac{k}{m} \cdot \opt$ in the worst case.  Consider the following instance on a line and choice given by MNL model. Let $\theta_j=1/m$ be the probability of  visiting $(1:j)$. 
Consider $N$ a universe of $m$ distinct products such that for all $i \in N$,  $r_i=m$  and $v_i=1/m$. Clearly $S_k^*=\{1,\ldots,k\}$ and $S_m^*=N$. 
We have
$R(S_k^*)=k/(1+k/m)$ and $R(S_m^*)=m/2$. 
Let $X_1$ be a placement of products $\{1,\ldots,k\}$ on the line. The expected  revenue of placement $X_1$ is given by 
\begin{align*}
	\mathcal{R}(X_1) = \sum_{j=1}^m \theta_j R(X_1(1:j))   = \sum_{j=1}^k \theta_j R(X_1(1:j)) +  \sum_{j=k+1}^m \theta_j R(X_1(1:j)).
\end{align*}
We have for any $j \leq k$, $R(X_1(1:j)) \leq  R(\{1,\ldots,j\})=j/(1+j/m) \leq j$. We have for any $   j \geq k+1$, $R(X_1(1:j)) \leq  R(\{1,\ldots,k\})=k/(1+k/m) \leq k$. Therefore,
\begin{align*}
	\mathcal{R}(X_1)  \leq \frac{1}{m} \left( \sum_{j=1}^k j +  \sum_{j=k+1}^{m} k    \right)= \frac{k(k+1)}{2m}+ \frac{(m-k)k}{m}= O(1) k.
\end{align*}

On the other hand, by considering $X_2$, the placement where product $j \in N$ is placed at location $j$ for every $j\in N$, we get 

$$\mathcal{R}(X_2)= \sum_{j=1}^m \theta_j R(\{1,\ldots,j\}) = 1/m \sum_{j=1}^m \frac{j}{1+j/m} \geq 1/m \sum_{j=1}^m \frac{j}{2} =\frac{m+1}{4}.$$
Therefore, for any placement $X_1$, we have $\mathcal{R}(X_1) = O(1) \cdot \frac{k}{m}   \cdot \mathcal{R}(X_2)   \leq O(1) \cdot \frac{k}{m} \cdot \opt $. 

\hfill\Halmos		
\end{proof}

\subsection{Proof of Lemma \ref{helper}} \label{apx:lem:helper}
\begin{proof}{Proof of Lemma \ref{helper}.} 
Let $\onee(i\in S^r_j)$ be an indicator random variable that is 1 if $i\in S^r_j$ and 0 otherwise. Using weak-rationality of $\phi$ (Assumption \ref{ration}), we have,
\[\ep[R(S^r_j)]=\sum_{i\in S^*_j}r_i \ep[\onee(i\in S^r_j)\,\phi(i,S^r_j)]\geq \sum_{i\in S^*_j} r_i \ep[\onee(i\in S^r_j)]\phi(i,S^*_j).\]
We claim that, $\ep[\onee(i\in S^r_j)]= 1-\left(1-\frac{1}{j}\right)^j$. Note that $\ep[1-\onee(i\in S^r_j)]$ is the probability that $i$ is never sampled when we independently sample (with replacement) $j$ products from $S^*_j$. Observe that,
\[\ep[1-\onee(i\in S^r_j)]=\left(1-\frac{1}{j}\right)^j,\]
here $1-\frac{1}{j}$ is the probability that a randomly sampled product from $S^*_j$ is not $i$. Note that we have equality due to the fact that $|S^*_j|=j$ (as discussed in Section \ref{sec:oracle}). Finally, note that $$1-\left(1-\frac{1}{j}\right)^j\geq 1-\frac{1}{e}\qquad \forall j\geq 1.$$ This completes the proof.
\hfill\Halmos		\end{proof}


\subsection{Proof of Lemma \ref{randline}} \label{apx:lem:randline}
\begin{proof}{Proof of Lemma \ref{randline}.}		

The proof mimics the main steps in the proof of Lemma \ref{line}. The main difference is that we lose a constant factor in the guarantee since some products may be repeated due to randomized placement (Lemma \ref{helper}). Let $O:G\to N$ denote an optimal placement, i.e., $\mathcal{R}(O)=\opt$. Let $O(j_1:j_2)$ denote the set of products placed in locations $\{j_1,j_1+2,\cdots, j_2\}$. 
{  For any given $\ell\in[\log_2 (\rho)]$, consider the placement solution where we place products $O(2^{\ell-1}j_{\min}:2^{\ell}j_{\min}-1)$ between $2^{\ell-1}j_{\min}$ and $2^\ell j_{\min} -1$, same as $O$, and keep all other positions vacant. Let $O_{\ell}$ denote this solution.
	From the proof of Lemma~\ref{line}, we have,
	\[\mathcal{R}(O)\leq \sum_{\ell=1}^{\log_2 (\rho)} \mathcal{R}(O_\ell).\]
	Consider some $\ell  \in [\log_2 (\rho)]$. As in Lemma~\ref{line}, we define the subset of locations  $G_{\ell,k}$ as 
	\[G_{\ell,k}= (2^{\ell-1}j_{\min}:\min\{k,2^{\ell}j_{\min}-1\}) .\]
	The total revenue of random placement $X^r_{2^{\ell-1} j_{\min}}$ can be lower bounded by the revenue of placement $O_\ell$ as follows,
	\begin{eqnarray*}
		\mathcal{R}(O_\ell)&= &\sum_{j=2^{\ell-1}j_{\min}}^{2^{\ell}j_{\min}-1}r_{O(j)} \sum_{k=j}^{j_{\max}} \theta_k \phi\left(O(j), O(G_{\ell,k}) \right),\\
		&= & \sum_{k= 2^{\ell-1}j_{\min}}^{j_{\max}} \theta_k\, R(O(G_{\ell,k})),\\
		&\leq& \frac{1}{\alpha}\sum_{k= 2^{\ell-1}j_{\min}}^{j_{\max}} \theta_k\, R\left(S^*_{2^{\ell-1}j_{\min}}\right),\\
		&\leq & \frac{1}{\alpha(1-1/e)} \sum_{k= 2^{\ell-1}j_{\min}}^{j_{\max}} \theta_k\, \ep[R\left(S^r_{2^{\ell-1}j_{\min}}\right)],\\
		&=&\frac{1}{\alpha(1-1/e)} \ep[\mathcal{R}(X^r_{2^{\ell-1} j_{\min}})],
	\end{eqnarray*}
	where the first inequality follows from \eqref{eq:approx} and the second inequality follows from Lemma \ref{helper}.
	Thus,
	\[\alpha(1-1/e) \mathcal{R}(O)\leq \sum_{\ell=1}^{\log_2(\rho)}\ep[\mathcal{R}(X^r_{2^{\ell-1}j_{\min}})].\]}
\hfill\Halmos		\end{proof}

\subsection{Proof of Theorem \ref{thm:tigh}}\label{appx:tight}


To prove Theorem \ref{thm:tigh} we provide the following three technical lemmas. 

\begin{lemma} \label{lem:Skstar}
For each $k \in [m]$, the optimal assortment $S^*_k$ of size at most $k$ in our instance is the set of the $k$ products with preference weights $v_k$. 
\end{lemma}

\begin{lemma} \label{lem:first}
Let $k \in [m]$. Consider $X_k$ a placement of the products in $S^*_k$  to the locations of the line (possibly a randomized placement). Let $ \mathcal{R}(X_k)$ be the revenue obtained by this placement. Then, $ \mathcal{R}(X_k)=O(1)$.
\end{lemma}

\begin{lemma} \label{lem:second}
Let $ X$ be a placement that, for each $j \in [m]$, assigns the product of preference weight $u_j$ to location $j$. Then, $ \mathcal{R}(X)=\Omega(\log^{1-\epsilon} m)$.
\end{lemma}

 Lemma \ref{lem:Skstar} characterizes the optimal assortments $S^*_k$ for each $k$.  Lemma \ref{lem:first} shows that every placement of these assortments has constant revenue. Lemma \ref{lem:second} shows that a placement solution made of the products that are not included in $\cup_{k\in[m]}S^*_k$ has substantially higher revenue. 
Combining the results of Lemma \ref{lem:first} and \ref{lem:second} directly yields Theorem~\ref{thm:tigh}. The proofs of the three lemmas are presented below.

\subsubsection{Proof of Lemma \ref{lem:Skstar}}

\begin{proof}{Proof of Lemma \ref{lem:Skstar}.}

	Consider an  assortment $S$ from our universe with cardinality $k$ and let $\beta_1, \ldots, \beta_k$ be the preferences weights of  products in $S$. We have,
	$$R(S)=  \frac{\sum_{i=1}^k \beta_i^{\frac{\epsilon}{1+\epsilon}}}{1+ \sum_{i=1}^k \beta_i}.    $$
	Note that $S$ can also have less than $k$ products by setting $\beta_i=0$.
	By the weighted AM-GM inequality, we have
	$$ \beta_i + \frac{1}{k} =              \frac{1+\epsilon}{\epsilon} \left(     \frac{\epsilon}{1+\epsilon} \cdot \beta_i + \frac{1}{1+\epsilon} \cdot \frac{\epsilon}{k}              \right)  \geq 
	\frac{1+\epsilon}{\epsilon} \left(     \beta_i^{ \frac{\epsilon}{1+\epsilon}}     \left( \frac{\epsilon}{k}\right)^{\frac{1}{1+\epsilon} }             \right).
	$$
	By taking the sum over $i$ from $1$ to $k$, we get
	$$ R(S) \leq   \frac{\epsilon}{1+\epsilon}   \left( \frac{k}{\epsilon}\right)^{\frac{1}{1+\epsilon} }     .               $$
	The upper bound is tight for $S_k^*$ because by letting $S^*_k$ be the set of the $k$ products with preference weights $v_k$, we have
	$$ R(S_k^*)= \frac{k v_k^{\frac{\epsilon}{\epsilon+1}}}{1+k v_k} =  \frac{k (\epsilon/k)^{\frac{\epsilon}{1+\epsilon}}}{1+ \epsilon } = \frac{\epsilon}{1+\epsilon}   \left( \frac{k}{\epsilon}\right)^{\frac{1}{1+\epsilon} }       .   $$
	\hfill\Halmos		\end{proof}

\subsubsection{Proof of Lemma \ref{lem:first}}


\begin{proof}{Proof of Lemma \ref{lem:first}.}
	Let $S_j$ be the assortment of products that the placement $X_k$ assigns to locations $\{1,\ldots,j\}$. Note that $S_j$ can be random. We have
	$$  \mathcal{R}(X_k)= \sum_{j=1}^m \theta_j \mathbb{E}(R(S_j)).$$
	Note that  $S_j \subseteq S_k^*$ where $S_k^*$ is given by Lemma \ref{lem:Skstar}. Therefore,
	$$\mathbb{E}(R(S_j)) =\mathbb{E}     \left(    \frac{| S_j | }{1+ | S_j | \epsilon /k }  \left(   \frac{\epsilon}{k}\right)^{\frac{\epsilon}{1+\epsilon}} \right) \leq
	\left(    \frac{\min(j,k) }{1+ \min(j,k)  \epsilon /k }  \left(   \frac{\epsilon}{k}\right)^{\frac{\epsilon}{1+\epsilon}} \right)
	\leq     \lambda \min (j,k) k^{-\frac{\epsilon}{1+\epsilon}}
	$$
	where $\lambda$ is the constant $\epsilon^{\frac{\epsilon}{1+\epsilon}}$.
	Therefore,
	$$\sum_{j=1}^k \theta_j \mathbb{E}(R(S_j)) \leq  \lambda k^{-\frac{\epsilon}{1+\epsilon}} \sum_{j=1}^k j \theta_j =  \lambda k^{-\frac{\epsilon}{1+\epsilon}} \sum_{j=1}^k j^{\frac{-1}{1+\epsilon}} = O(1).  $$
	Moreover, 
	$$\sum_{j=k+1}^m \theta_j \mathbb{E}(R(S_j)) \leq  \lambda k^{-\frac{\epsilon}{1+\epsilon}} k \sum_{j=k+1}^m  \theta_j =  \lambda k^{\frac{1}{1+\epsilon}} \sum_{j=k+1}^m j^{-1-\frac{1}{1+\epsilon}} = O(1).  $$
	Adding the last two equations give the desired result.
	\hfill\Halmos		\end{proof}

\subsubsection{Proof of Lemma \ref{lem:second}}


\begin{proof}{Proof of Lemma \ref{lem:second}.}
	Let $S_j$ be the assortment of products that $X$ places  in locations $\{1,\ldots,j\}$. 
	We have $$R(S_k)=  \frac{\sum_{j=1}^k u_j^{\frac{\epsilon}{1+\epsilon}}}{1+ \sum_{j=1}^k u_j} .   $$
	We have $$\sum_{j=1}^k u_j = \sum_{j=1}^k  1/( j \log^{1+\epsilon} j) = O(1) (\log_2 k)^{-\epsilon} ,    $$ hence $1+\sum_{j=1}^k u_j =O(1)$.
	On the other hand, we have $$\sum_{j=1}^k u_j^{\frac{\epsilon}{1+\epsilon}} \geq k  u_k^{\frac{\epsilon}{1+\epsilon}} = k^{\frac{1}{1+\epsilon}} (\log_2 k)^{-\epsilon}.
	$$
	Therefore, 
	$$R(S_k)= \Omega(1) k^{\frac{1}{1+\epsilon}} (\log_2 k)^{-\epsilon}.$$
	Finally,
	$$ \mathcal{R}(X)=  \sum_{j=1}^m \theta_j R(S_j) =   \Omega(1) \sum_{j=1}^m j^{-1-\frac{1}{1+\epsilon}}         j^{\frac{1}{1+\epsilon}} (\log_2 j)^{-\epsilon}    =    \Omega(1) \sum_{j=1}^m j^{-1}         (\log_2 j)^{-\epsilon}=  \Omega(\log^{1-\epsilon} m).  $$
	
	\hfill\Halmos		\end{proof}

\subsection{Missing Proofs from Section \ref{subsec:deran}}\label{appx:deran}

 Our derandomized algorithm is based on properties of optimal unconstrained assortments for Markov choice models (\cite{desir, SO}). In particular, we use the notion of \emph{compatibility} given in \cite{SO}. A choice model $\phi$ on ground set $N$ with optimal unconstrained assortment $S$ ($\subseteq N$) is compatible if,
\begin{eqnarray}
	R(A\cup C)-R(A)&\geq &0 \qquad  \forall  A\subseteq N,\, C\subseteq S, \label{prop2}\\
	R(A\cup C)-R(A)&\leq &R(B\cup C)-R(B) \qquad \forall B\subseteq A\subseteq S,\, C\subseteq N. \label{prop1}
\end{eqnarray}

\begin{lemma}[Part of Theorem 7 in \cite{SO}.]\label{compat}
	Every Markov choice model is compatible.
\end{lemma}

\subsubsection{Proof of Lemma  \ref{lem:redsubmod}}
\begin{proof}{Proof of Lemma  \ref{lem:redsubmod}.} Fix a $k\in[m]$ and let
	\[N_k=\argmax_{A\subseteq S^*_k} R(A).\]
	Given $S^*_k$, one can compute $N_k$ by solving an unconstrained assortment optimization problem. Recall that this can be done in polynomial time for the Markov choice model \cite{blanchet2016markov}. Now, we restrict our attention to $N_k$ and use it as a (reduced) ground set for the rest of the proof. By definition,
	\[N_k=\argmax_{A\subseteq N_k} R(A).\]
	
	Similar to the proof of Theorem \ref{uniform}, 
	consider the expanded ground set $N_k\times G=\{(i,j)\mid i\in N_k, j\in G\}$. A placement $X$ corresponds to set $\left\{(X(j),j)\mid j\in G \right\}\subseteq N_k\times G$. Given a set $U\subseteq N_k\times G$, for each $j \in G$, we define the following set of products
	\[U(j)=\{  i\in N_k    \mid (i,j)\in U \}.\] 
	Consider the matroid with family of independent sets $\mathcal{F}=\{U\mid\, |U(j)|\leq 1\,\, \forall j\in G\}$. Each placement $X$ corresponds to a unique set in $\mathcal{F}$. Similarly, given a set $U\in \mathcal{F}$, we have a unique placement solution where location $j\in G$ has product $i\in N_k$ if $(i,j)\in U$ for some $i$, and $j$ has no product otherwise.
	Finally, given an arbitrary set $U\subseteq N_k\times G$, let
	\begin{equation*}
		\mathcal{R}(U)=\sum_{L \in 2^G}  \mathbb{P}_B(L) \,\, R\left(\cup_{j \in L} U(j)\right).
	\end{equation*}
	
	The placement problem is equivalent to the following matroid constrained subset selection problem on ground set $N_k\times G$, 
	\begin{equation}\label{spot}
		\max_{U\in \mathcal{F}} \mathcal{R}(U).
	\end{equation}
	To complete the proof, we show that $R(\cdot)$ is monotone and submodular on the ground set $N_k$. Similar to the proof of Theorem \ref{uniform}, this establishes the monotonicity and submodularity of $\mathcal{R}(\cdot)$ on the ground set $N_k\times G$ and reduces \eqref{spot} to an instance of matroid constrained monotone submodular function maximization. Existing algorithms~(\cite{calinescu2011maximizing}) give a $(1-1/e)$-approximation to this problem. 
	
	To show that $R(\cdot)$ is monotone and submodular, we use Lemma \ref{compat} with both $N$ and $S$ replaced by $N_k$ to get,
	\begin{eqnarray*}
		R(A\cup\{e\})-R(A)&\geq &0 \qquad  \forall A\subseteq N_k,\, e\in N_k,\\
		R(A\cup\{e\})-R(A)&\leq &R(B\cup\{e\})-R(B) \qquad \forall B\subseteq A\subseteq N_k,\, e\in N_k, 
	\end{eqnarray*}
	here the first inequality follows from \eqref{prop2} and the second inequality follows from \eqref{prop1}. The first inequality establishes monotonicity and the second establishes submodularity of $R(\cdot)$ on the ground set $N_k$, as desired.
	\hfill\Halmos	\end{proof}


\subsubsection{Proof of Theorem  \ref{thm:morocco}}
\begin{proof}{Proof of Theorem  \ref{thm:morocco}.}
	Let $\hat{O}_k$ denote the optimal placement of $S^*_k$. From Lemma \ref{lem:redsubmod}, we have, 
	\[\mathcal{R}(\hat{Y})\geq     \max_{k\in[m]} \mathcal{R}(\hat{X}_k) \geq   (1-1/e) \ \cdot \max_{k\in[m]} \mathcal{R}(\hat{O}_k).\]
	Theorem \ref{randgen} states that there exists a placement of one of the assortments from the collection $\{S^*_k\}_{k\in[m]}$, that is $\alpha\frac{1-1/e}{\log_2 \rho} $- approximate. Thus,
	\[\max_{k\in[m]} \mathcal{R}(\hat{O}_k)\geq \alpha\frac{1-1/e}{\log_2 \rho} \opt,\] 
	where $\opt$ is the optimal objective value of \place.
	Combining these inequalities, we have that
	$\hat{Y}$ is a $\alpha\frac{(1-1/e)^2}{\log_2 \rho}$ approximation of the optimal placement solution. Recall that, $\alpha=0.5-\epsilon$ for the Markov choice model~(\cite{desir,SO}). Therefore, we get $\frac{\Omega(1)}{\log_2 \rho}$-approximation (deterministic) algorithm for \place.
	
	\hfill\Halmos
\end{proof}

	{
	\section{Relaxing Assumption \ref{decouple}}\label{appx:decouple}
	Assumption \ref{decouple} states that the browsing distribution $B$ and choice model $\phi$ are both independent of the placement $X$. In this section, we partially relax this assumption and consider a setting 
	where product placement influences customer choice even after the customer has determined his/her consideration set, i.e., the choice model depends on the placement $X$ but the browsing distribution is independent of $X$. We consider placement dependent MNL choice model and show that our algorithm and its approximation guarantee seamlessly generalizes in this setting.  

\subsubsection*{Model Formulation.} 
Our goal is to capture the phenomenon that customers may be more likely to purchase a product that they saw earlier on their path even after they have formed their consideration set~(\cite{chandon2009does, gallego2020approximation}). To this end, we consider a generalization of MNL choice model where the attraction parameter of every product is a function of \emph{when} the customer first saw the product on their path through the store. Specifically, we have a set of attraction parameters $\{v_{i,k}\}_{k\in[m]}$ for every product $i\in N$ where $v_{i,k}$ is the attraction parameter of $i$ when the customer sees $k$ ($\geq 0$) distinct products before the first time that they see product $i$ on their path. We assume that,
\begin{equation}\label{assumemnl}
	v_{i,k}\geq v_{i,k+1}\quad	 \forall k\geq 0,\, i\in N.
\end{equation}
This is a reasonable assumption for capturing customers' tendency to give higher valuation to products that are seen earlier on the path. Note that $v_{i,k}=v_{i,k+1}\,\, \forall k\geq 0$ for the standard MNL model. A special case of the model (and the assumption) above were introduced by \cite{gallego2020approximation} for the \placel\ problem.

Recall that, we can describe a customer's path through the store by a permutation $\pi$ over the set of location $[m]$ and a stopping location $j$. Given a placement $X$, consider a customer who visits the $j$ locations $\pi(1),\pi(2),\cdots,\pi(j)$ (in order) and sees assortment $X^{\pi}(1:j)$. For every product $i\in X^{\pi}(1:j)$, let $p^{\pi}_i(X)$ denote the number of distinct products placed before product $i$ on the path $\pi(1),\cdots,\pi(j)$. The expected revenue of this customer is given by,
\[R^{\pi}(X^{\pi}(1:j))=\frac{\sum_{i\in X^{\pi}(1:j)}r_iv_{i,p^{\pi}_i(X)}}{1+\sum_{i\in X^{\pi}(1:j)}v_{i,p^{\pi}_i(X)}},\]
here the 1 in the denominator represents the attraction parameter of the outside option. We use notation $R^{\pi}$ to embellish the fact that the revenue function depends on the permutation $\pi$. 
Now, the expected total revenue of placement $X$ is given by
\begin{eqnarray*}
	\mathcal{R}(X)&=& \sum_{\pi\in \Pi}\sum_{j\in [m]} \theta^{\pi_j}\, R^{\pi}(X^{\pi}(1:j)).
\end{eqnarray*}
Recall that our goal  is to find a placement with the highest expected revenue. 
We refer to this problem as \place+, 
\begin{equation*} \label{place1}
	\opt=\qquad\max_{X\in \mathcal{X}}\, \mathcal{R}(X)
\end{equation*}
When $v_{i,k}=v_{i,k+1}\,\, \forall k\geq 0,i\in N$, we recover the original \place\ formulation with MNL choice model.

To introduce the algorithm for \place+, we start by generalizing some notation and basic assumptions from Section \ref{sec:prelim}.
We assume an oracle that outputs the expected revenue $\mathcal{R}(X)$ for any placement solution $X$. Similar to the original formulation, w.l.o.g., we allow some locations to remain empty.  For every $j\in[m]$, let $\phi_j$ denote the MNL choice model with attraction parameters $\{v_{i,j}\}_{i\in[n]}$ and let 
\[R_j(S)=\frac{\sum_{i\in S} r_i\,v_{i,j}}{1+\sum_{i\in S} v_{i,j}}\quad \forall S\subseteq N.\]
Recall that the optimal cardinality constrained assortment can be computed in polynomial time for MNL choice model (\cite{rusmevichientong2010dynamic}). Let $S^*_{k}$ denote the optimal assortment of at most $k$ products for choice model $\phi_k$. W.l.o.g., $|S^*_{k}|=k$.  Formally,
\begin{equation} \label{eq:approx1}
	R_k(S^*_{k})= \max_{|S|=k}R_k(S).
\end{equation}

\subsubsection*{Algorithm.} 
Given assortment $S^*_{k}$ as defined in Equation \eqref{eq:approx1}, let $X^r_k$ denote a uniformly random placement of products in $S^*_{k}$, i.e., at each location we place a uniformly randomly (and interdependently) chosen product from $S^*_{k}$. 
Let, 
\[Y=\argmax_{X\in \{X^r_k\}_{k \in [m]}} \mathcal{R}(X),\]
denote the final (randomized) placement solution. 



\subsubsection*{Approximation Guarantee.}

We establish the following approximation guarantee for the algorithm.  Recall that, $j_{\min}$ ($\geq 1$) and $j_{\max}$ ($\leq m$) denote the lengths of the shortest and longest paths with non-zero support in the browsing distribution, and $\rho=\frac{j_{\max}+1}{j_{\min}}$ ($\leq m+1$). W.l.o.g., $\rho=\frac{j_{\max}+1}{j_{\min}}=2^d$ for some integer $d\in \{1,2,\cdots,{\log_2 (m+1)}\}$.

\begin{theorem}\label{randgen1}
	The randomized placement $Y$ gives $\frac{\Omega(1)}{\log \rho}$-approximation for \place+, i.e.,  
	\[\ep[\mathcal{R}(Y)]=\ep[\max_{k \in [m]} \mathcal{R}(X^r_k)]\geq \,\frac{1-1/e}{\,\log_2 \rho} \opt.\]	
\end{theorem} 

Similar to the proof of Theorem \ref{randgen}, we first show the result for a line graph (Lemma \ref{randline1}) and then use the result for a line graph as a building block for establishing Theorem \ref{randgen1}. We start by stating and proving two lemmas (\ref{mnl} and \ref{helper1}) that will help us prove the result for a line.	
\begin{lemma}\label{mnl}
	Let $\phi_1$ and $\phi_2$ be MNL choice models with attraction parameters $\{v_{i,1}\}_{i\in N}$ and $\{v_{i,2}\}_{i\in N}$ respectively such that $v_{i,1}\geq v_{i,2}\,\, \forall i\in N$. For $\ell\in\{1,2\}$, let $R_\ell$ denote the revenue function for $\phi_\ell$ and let $Z_\ell=\argmax_{|S|\leq k} R_\ell(S)$ for some $k\in[n]$. Then,
	\begin{enumerate}[$(i)$]
		\item For every $Z\subset Z_2$ and $e\in Z_2\backslash Z$, we have, $R_2(Z\cup\{e\})\geq R_2(Z)$.
		\item Assortment $Z_2$ generates more revenue under choice model $\phi_1$, i.e., $R_1(Z_2)\geq R_2(Z_2)$.
		\item The optimal revenue in $\phi_1$ is higher than in $\phi_2$, i.e., $R_1(Z_1)\geq R_2(Z_2)\geq R_2(Z)\,\, \forall Z\subseteq N$.
	\end{enumerate}
\end{lemma}
\begin{proof}{Proof.}
	To prove claims $(i)$ and $(iii)$, we use the following property of optimal assortments under MNL choice model: the price of every product in the optimal assortment $Z_2$ is at least as much as the expected revenue of the assortment, i.e., for all $i \in Z_2$, 
	\begin{equation}\label{helper2}
		r_i\geq R_2(Z_2).
	\end{equation} 
	We start by proving \eqref{helper2}.
	Consider an arbitrary element $e\in Z_2$. We have,
	\begin{eqnarray}
		R_2(Z_2)-R_2(Z_2\backslash \{e\})&= &\frac{\sum_{i\in Z_2} r_i v_{i,2}}{1+\sum_{i\in Z_2} v_{i,2}}-R_2(Z_2\backslash \{e\}),\nonumber\\
		&= &\frac{R_2(Z_2\backslash \{e\})+\frac{v_e r_e}{1+\sum_{i\in Z_2\backslash \{e\}} v_{i,2}}}{1+\frac{v_e}{1+\sum_{i\in Z_2\backslash \{e\}} v_{i,2}}}-R_2(Z_2\backslash \{e\}),\nonumber\\
		&= &\frac{v_e}{1+\sum_{i\in Z_2} v_{i,2}}\left(r_e- R_2(Z_2\backslash \{e\})\right),\nonumber\\
		R_2(Z_2)&=&  \left(1-\frac{v_e}{1+\sum_{i\in Z_2} v_{i,2}}\right)\,R_2(Z_2\backslash \{e\}) + \left(\frac{v_e}{1+\sum_{i\in Z_2} v_{i,2}}\right)\, r_e.\label{convexcomb1}
	\end{eqnarray}
	From \eqref{convexcomb1}, we have that $R_2(Z_2)$ is a convex combination of $r_e$ and $R_2(Z_2\backslash \{e\})$. Since $Z_2$ is an optimal assortment, we have $R_2(Z_2)- R_2(Z_2\backslash \{e\})\geq 0$. 
	Thus, $r_e\geq R_2(Z_2)$.
	\smallskip
	
	\noindent \emph{Proof of Claim $(i)$:}  From \eqref{helper2}, we have $r_e\geq R_2(Z_2)$. Further, $R_2(Z_2)\geq R_2(Z)$ because $Z_2\supset Z$ and $Z_2$ is an optimal assortment. We have,
	\begin{eqnarray}
		R_2(Z\cup\{e\})-R_2(Z)&= &\frac{\sum_{i\in Z} r_i v_{i,2} +r_ev_{e,2}}{1+\sum_{i\in Z} v_{i,2} + v_{e,2}}-R_2(Z),\nonumber\\
		&= &\frac{R_2(Z)+\frac{v_e r_e}{1+\sum_{i\in Z} v_{i,2}}}{1+\frac{v_e}{1+\sum_{i\in Z} v_{i,2}}}-R_2(Z),\nonumber\\
		&= &\frac{v_e}{1+\sum_{i\in Z} v_{i,2} +\, v_{e,2}}\left(r_e- R_2(Z)\right),\nonumber\\
		&\geq & 0.\nonumber
	\end{eqnarray}

	\noindent \emph{Proof of Claims $(ii)$ and $(iii)$:}	First, note that the third claim follows from the second because,
	\[R_1(Z_1)\geq R_1(Z_2)\geq R_2(Z_2).\]	
	To show the second claim, let $\Delta_i= v_{i,1}-v_{i,2}\,\, \forall i\in N$. We have,
	\begin{eqnarray*}
		R_1(Z_2)&= &\frac{\sum_{i\in Z_2} r_i v_{i,1}}{1+\sum_{i\in Z_2} v_{i,1}},\\
		&=& \frac{\sum_{i\in Z_2} r_i (v_{i,2}+\Delta_i)}{1+\sum_{i\in Z_2} (v_{i,2}+\Delta_i)},\\
		&=& \frac{R_2(Z_2)+\frac{\sum_{i\in Z_2}r_i\,\Delta_i}{1+\sum_{i\in Z_2} v_{i,2}}}{1+\frac{\sum_{i\in Z_2} \Delta_i}{1+\sum_{i\in Z_2} v_{i,2}}},\\
		&\geq & \frac{R_2(Z_2)+R_2(Z_2)\frac{\sum_{i\in Z_2}\Delta_i}{1+\sum_{i\in Z_2} v_{i,2}}}{1+\frac{\sum_{i\in Z_2} \Delta_i}{1+\sum_{i\in Z_2} v_{i,2}}},\\
		&=& R_2(Z_2).
	\end{eqnarray*}
	Here, the inequality follows from \eqref{helper2}. 
	
	\hfill\Halmos\end{proof}

\begin{lemma}\label{helper1}
	Given the assortment $S^*_{k}$ for $k\in[m]$, let $S^r_{k}$ denote  a random assortment of $k$ products where each product is independently sampled from $S^*_{k}$ (with replacement). We have, $\ep[R_k(S^r_{k})] \geq (1-1/e) R_k(S^*_{k}).$
\end{lemma}
\begin{proof}{Proof.} 
	Except for some differences in notation, the proof is identical to the proof of Lemma~\ref{helper}.	Let $\onee(i\in S^r_{k})$ be an indicator random variable that is 1 if $i\in S^r_{k}$ and 0 otherwise. Using weak-rationality of $\phi$ (Assumption \ref{ration}), we have,
	\[\ep[R_j(S^r_{k})]=\sum_{i\in S^*_{k}}r_i \ep[\onee(i\in S^r_{k})\,\phi_j(i,S^r_{k})]\geq \sum_{i\in S^*_{k}} r_i \ep[\onee(i\in S^r_{k})]\phi_j(i,S^*_{k}).\]
	Same as in the proof Lemma \ref{helper}, we have, $\ep[\onee(i\in S^r_{k})]= 1-\left(1-\frac{1}{k}\right)^k$. Now, using the fact that $1-\left(1-\frac{1}{k}\right)^k\geq 1-\frac{1}{e}\,\, \forall k\geq 1,$ completes the proof.
	\hfill\Halmos		\end{proof}
\begin{lemma}\label{randline1}
	For a line graph,	$
	\sum_{\ell=1}^{\log_2(\rho)}\ep[\mathcal{R}(X^r_{2^{\ell-1}j_{\min}})]\geq (1-1/e) \opt$.
\end{lemma}
\begin{proof}{Proof.}		
	
	On a line graph the order $\pi$ in which locations can be visited is fixed and the browsing distribution reduces to a distribution $\{\theta_j\}_{j\in[m]}$ over the last location visited by the customer. Let $O:G\to N$ denote an optimal placement, i.e., $\mathcal{R}(O)=\opt$. Let $O(j_1:j_2)$ denote the set of products placed in locations $\{j_1,j_1+2,\cdots, j_2\}$. Let $O(j)$ denote the product placed at location $j$. Because we are on a line, w.l.o.g., the optimal placement has a distinct product in every location and there are $j-1$ distinct products placed before the product $O(j)$.
	
	{ For any given $\ell\in[\log_2 (\rho)]$, consider the placement solution where we place products $O(2^{\ell-1}j_{\min}:2^{\ell}j_{\min}-1)$ between $2^{\ell-1}j_{\min}$ and $2^\ell j_{\min} -1$, same as $O$, and keep all other positions vacant. Let $O_{\ell}$ denote this solution. We have,
		\begin{eqnarray*}
			\mathcal{R}(O)&=&\sum_{\ell=1}^{\log_2 \rho} \sum_{j=2^{\ell-1} j_{\min}}^{2^{\ell}j_{\min}-1}r_{O(j)}\,\sum_{k\geq j} \theta_k \,\frac{v_{O(j),j}}{1+\sum_{i\in [k]}v_{O(i),i}},\\
			&\leq &\sum_{\ell=1}^{\log_2 \rho} \sum_{j=2^{\ell-1}j_{\min}}^{2^{\ell}j_{\min}-1}r_{O(j)} \sum_{k=j}^{j_{\max}} \theta_k \frac{v_{O(j),j}}{1+\sum_{i= 2^{\ell-1}j_{\min} }^{\min\{k,2^{\ell}j_{\min}-1\}}v_{O(i),i}},\\
			&=&\sum_{\ell=1}^{\log_2 (\rho)} \mathcal{R}(O_\ell).
		\end{eqnarray*}
		The inequality follows from the fact that $\sum_{i= 2^{\ell-1}j_{\min} }^{\min\{k,2^{\ell}j_{\min}-1\}}v_{O(i),i}\leq \sum_{i\in [k]}v_{O(i),i}$.
		
		Now, consider the randomized placement $X^r_{2^{\ell-1}j_{\min}}$, or $X^r_\ell$ for short. We obtain this placement by randomly placing a product from $S^*_{2^{\ell-1}j_{\min}}$ at every location. Let $X^r_\ell (1:k)$ denote the assortment of products in the first $k$ locations in the placement $X^r_\ell$. Note that $X^r_\ell (1:2^{\ell-1}j_{\min})$ is stochastically equivalent to an assortment generated by drawing $2^{\ell-1}j_{\min}$ i.i.d.\ samples from the uniform distribution over products in $S^*_{2^{\ell-1}j_{\min}}$. Let $p_i(X^r_\ell)$ denote the number of distinct products that are placed before product $i$ in $X^r_\ell$. Then, the total revenue of random placement $X^r_{2^\ell j_{\min}}$ can be lower bounded by $\mathcal{R}(O_\ell)$ as follows,
		\begin{eqnarray*}
			\mathcal{R}(O_\ell)&= &\sum_{j=2^{\ell-1}j_{\min}}^{2^{\ell}j_{\min}-1}r_{O(j)} \sum_{k=j}^{j_{\max}} \theta_k \frac{v_{O(j),j}}{1+\sum_{i= 2^{\ell-1}j_{\min} }^{\min\{k,2^{\ell}j_{\min}-1\}}v_{O(i),i}},\\
			&= & \sum_{k= 2^{\ell-1}j_{\min}}^{j_{\max}} \theta_k\, \sum_{j=2^{\ell-1}j_{\min}}^{\min\{k,2^{\ell}j_{\min}-1\}}\frac{r_{O(j)}\,v_{O(j),j}}{1+\sum_{i= 2^{\ell-1}j_{\min} }^{\min\{k,2^{\ell}j_{\min}-1\}}v_{O(i),i}},\\
			&\leq& \sum_{k= 2^{\ell-1}j_{\min}}^{j_{\max}} \theta_k\, R_{2^{\ell-1}j_{\min}}\left(S^*_{2^{\ell-1}j_{\min}}\right),\\
			&\leq & \frac{1}{(1-1/e)} \sum_{k= 2^{\ell-1}j_{\min}}^{j_{\max}} \theta_k\, \ep\left[R_{2^{\ell-1}j_{\min}}\left(X^r_\ell(1:2^{\ell-1}j_{\min})\right)\right],\\
			&\leq & \frac{1}{(1-1/e)} \sum_{k= 2^{\ell-1}j_{\min}}^{j_{\max}} \theta_k\,\ep\left[ \sum_{e\in X^r_\ell(1:2^{\ell-1}j_{\min})}\frac{r_e\,v_{e,p_e(X^r_\ell)}}{1+\sum_{i\in X^r_\ell(1:2^{\ell-1}j_{\min})}v_{i,p_{i}(X_\ell^r)}}\right],\\
			&\leq & \frac{1}{(1-1/e)} \sum_{k= j_{\min}}^{j_{\max}} \theta_k\,\ep\left[ \sum_{e\in X^r_\ell(1:k)}\frac{r_e\,v_{e,p_e(X^r_\ell)}}{1+\sum_{i\in X^r_\ell(1:k)}v_{i,p_{i}(X_\ell^r)}}\right],\\
			&=&\frac{1}{(1-1/e)} \ep[\mathcal{R}(X^r_{2^{\ell-1} j_{\min}})],
		\end{eqnarray*}
		where the first inequality follows from our assumption \eqref{assumemnl} and Lemma \ref{mnl} $(iii)$. The second inequality follows from Lemma \ref{helper1}. The third inequality follows from assumption \eqref{assumemnl} and Lemma \ref{mnl} $(ii)$. The final inequality follows from Lemma \ref{mnl} $(i)$. 
		Thus,
		\[(1-1/e) \mathcal{R}(O)\leq \sum_{\ell=1}^{\log_2(\rho)}\ep[\mathcal{R}(X^r_{2^{\ell-1}j_{\min}})].\]}
	\hfill\Halmos		\end{proof}
{ Using the fact that, 
	\[\sum_{\ell=1}^{\log_2(\rho)}\ep[\mathcal{R}(X^r_{2^{\ell-1}j_{\min}})]\leq (\log_2 \rho)E[\mathcal{R}(Y)],\] 
	we have that $Y$ is a $\frac{(1-1/e)}{\log_2 \rho}$-approximation for a line graph as a direct corollary of Lemma \ref{randline1}.}

\begin{proof}{Proof of Theorem \ref{randgen1}.}
	Apart from minor differences in notation, we essentially mimic the proof of Theorem \ref{randgen}.
	Fix an arbitrary permutation $\pi:[m]\to[m]$ and let $W^\pi(X)$ denote the expected total revenue from placement $X$ when customers visit a subset of locations $\{\pi(1), \pi(2),\cdots, \pi(m)\}$ in order, i.e., 
	\[W^{\pi}(X)=\sum_{j\in[m]} \theta^{\pi}_jR^{\pi}(X^{\pi}(1:j)).\] 
	Let $W^{\pi}(O)$ denote the expected total revenue of the optimal placement when customers visit locations in the order given by $\pi$.	Note that fixing $\pi$ reduces the problem to an instance of \placel. 	From Lemma \ref{randline1}, we have 
	\[ \sum_{\ell=1}^{\log_2(\rho)}\ep[W^{\pi}(X^r_{2^{\ell-1}j_{\min}})]\geq  (1-1/e)W^{\pi}(O).\]	
	Now, the total revenue of placement $X$ can be written as,
	\[\mathcal{R}(X)=\sum_{\pi\in\Pi} W^{\pi}(X).\]
	Thus, for the randomized placement $Y$, we have,

	\begin{eqnarray*}
		(\log_2 \rho)\, \ep[\mathcal{R}(Y)]= (\log_2 \rho)\,\sum_{\pi\in \Pi}\ep[W^{\pi}(Y)]
		& \geq &\sum_{\pi\in \Pi}\sum_{\ell=1}^{\log_2 \rho}\ep[W^{\pi}(X^r_{2^{\ell-1}j_{\min}})],\\
		&\geq &\sum_{\pi\in \Pi}(1-1/e)\,W^{\pi}(O) = (1-1/e) \opt.		\end{eqnarray*}
	\hfill\Halmos		\end{proof}
}

\section{Mixed Integer Linear Program for \place\ under MNL} \label{appendix:IP}

In this section, we provide a MILP formulation for the \place\ problem for MNL choice model. We consider a general browsing distribution. Suppose the browsing distribution specifies a certain number of paths, say $\Gamma$, where each path $j \in [\Gamma]$ corresponds to a subset of locations visited with probability $\theta_j$. For each product $i \in [n]$ and each location $\ell \in [m]$, we use the binary variable $x_{i\ell}$ to indicate if product $i$ is placed in location $\ell$. Thus, the binary variables $(x_{i\ell})_{i\in[n],\, \ell\in[m]}$ represent the placement solution. For each product $i \in [n]$ and each path $j \in [\Gamma]$, we use the binary variable $y_{ij}$ to indicate if product $i$ is placed in at least one of the locations in path $j$ in our placement. The variables $x_{i\ell} \in \{0,1\}$ and $y_{ij} \in \{0,1\}$ are connected through the following constraints for all $i \in [n]$ and $j \in [\Gamma]$
\[
x_{i\ell} \leq y_{ij}, \quad \forall \ell \in {\sf path}(j), \quad \quad
\text{and} \quad \quad
y_{ij} \leq \sum_{\ell \in {\sf path}(j)} x_{i\ell}.
\]
The first constraint ensures that if product $i$ is placed in at least one location in path $j$ (i.e., $x_{i\ell} = 1$ for some $\ell \in {\sf path}(j)$), then $y_{ij} = 1$, implying product $i$ is visited in path $j$. Conversely, if product $i$ is not placed in any location in path $j$, then $\sum_{\ell \in {\sf path}(j)} x_{i\ell} = 0$, which implies $y_{ij} = 0$.

Moreover, we need to ensure that each location gets at most one product through the constraint $\sum_{\ell=1}^m x_{i \ell} \leq 1$ for all $ i \in [n].$

Let $v_i$ denote the preference weight of product $i$ under MNL and $r_i$ denote the price of product~$i$. Observe that the revenue from path $j$ is given by,
\[
\frac{\sum_{i=1}^n r_i v_i y_{ij}}{1 + \sum_{i=1}^n v_i y_{ij}}.
\]
Therefore, the objective function (total expected revenue) is:
\[
\sum_{j=1}^{\Gamma} \theta_j \frac{\sum_{i=1}^n r_i v_i y_{ij}}{1 + \sum_{i=1}^n v_i y_{ij}},
\]
where we sum over all paths weighted by their corresponding probabilities. This leads to the following formulation of the placement problem:
\begin{equation}
\begin{aligned}
& \max && \sum_{j=1}^{\Gamma} \theta_j \frac{\sum_{i=1}^n r_i v_i y_{ij}}{1 + \sum_{i=1}^n v_i y_{ij}} &\\
& \text{s.t.} && x_{i\ell} \leq y_{ij}, \quad& \forall i \in [n], \forall j \in [\Gamma], \forall \ell \in {\sf path}(j), \\
&&& y_{ij} \leq \sum_{\ell \in {\sf path}(j)} x_{i\ell}, \quad &\forall i \in [n], \forall j \in [\Gamma], \\
&&& \sum_{\ell=1}^m x_{i \ell} \leq 1, \quad &\forall i \in [n], \\
&&& y_{ij} \in \{0,1\}, \quad x_{i\ell} \in \{0,1\},   \quad & \forall i \in [n], \forall \ell \in [m], \forall j \in [\Gamma].
\end{aligned}
\end{equation}
We proceed to linearize the above problem as follows. First, we replace the revenue from each path $j$ with a new variable $z_j$ and move the fraction to the constraints. In particular, we get
\[
z_j + \sum_{i=1}^n v_i z_j y_{ij} \leq \sum_{i=1}^n r_i v_i y_{ij}, \quad \forall j \in [\Gamma].
\]
Note that the term $z_j y_{ij}$ is bilinear. Since $y_{ij}$ is binary, we linearize this term by introducing a new variable $w_{ij}$ and ensuring $w_{ij} = z_j y_{ij}$ through the following big-M constraints for all $i \in [n]$ and $j \in [\Gamma]$:
\[
0 \leq w_{ij} \leq M y_{ij}, \quad \text{and} \quad -M (1-y_{ij}) + z_j \leq w_{ij} \leq z_j,
\]
where $M$ is a sufficiently large number that bounds all $z_j$. It is sufficient to set $M = \max_{i \in [n]} r_i$.

To summarize, the MILP for the \place\ problem is:
\begin{equation}
\label{prob:ip}
\begin{aligned}
& \max && \sum_{j=1}^{\Gamma} \theta_j z_j \\
& \text{s.t.} && z_j + \sum_{i=1}^n v_i w_{ij} \leq \sum_{i=1}^n r_i v_i y_{ij}, \quad & \forall j \in [\Gamma], \\
&&& 0 \leq w_{ij} \leq M y_{ij}, \quad & \forall i \in [n], \forall j \in [\Gamma], \\
&&& -M (1-y_{ij}) + z_j \leq w_{ij} \leq z_j, \quad & \forall i \in [n], \forall j \in [\Gamma], \\
&&& x_{i\ell} \leq y_{ij}, \quad & \forall i \in [n], \forall j \in [\Gamma], \forall \ell \in {\sf path}(j), \\
&&& y_{ij} \leq \sum_{\ell \in {\sf path}(j)} x_{i\ell}, \quad & \forall i \in [n], \forall j \in [\Gamma], \\
&&& \sum_{\ell=1}^m x_{i \ell} \leq 1, \quad &\forall i \in [n], \\
&&& y_{ij} \in \{0,1\}, \quad x_{i\ell} \in \{0,1\}, \quad w_{ij} \geq 0, \quad z_j \geq 0, \quad  & \forall i \in [n], \forall \ell \in [m], \forall j \in [\Gamma].
\end{aligned}
\end{equation}

Note that our formulation is valid for general browsing distributions. We use it for both the line and the random walk in our numerical experiments in Section~\ref{sec:numerics}. In particular, for a line with $m$ locations, $\Gamma = m$, i.e., there are $m$ paths where path $j$ corresponds to locations $1$ to $j$ and has a probability of $\theta_j$. For the 2D random walk, we generate $\Gamma$ random walks, where each random walk gives the subsets of locations visited in the grid. Each random walk has uniform probability $\theta_j = \frac{1}{\Gamma}$.

{\color{black}
\section{Additional numerical experiments} \label{appx:exp}
In Figure~\ref{fig:occ-scores}, products are ranked by revenue, with product 1 being the most expensive, and the bars report the average number of occurrences of each product across all instances. We show six panels: three corresponding to ALG and three to OPT, each for a different value of $n$, with bar colors distinguishing the number of locations. Several patterns clearly emerge. First, the most duplicated product is consistently the highest-revenue product, and the number of occurrences decreases as the product rank decreases, producing a revenue-ordered structure. This behavior is consistent across all values of $n$ and $m$. Second, ALG systematically uses more duplicates than OPT, reflecting its randomized construction. We would like to note that  the plots are truncated at the first ten products for $n=15$ and $n=20$, because beyond that point the average occurrences are zero. This is a direct consequence of the revenue distribution: we generate revenues from an exponential distribution, so products in the tail have very small values and are rarely, if ever, used. Overall, these findings align with the well-known properties of the MNL model, where the optimal assortment tends to be revenue-ordered.  

Figure~\ref{fig:occ-scores} focuses solely on occurrences, independent of the probability that locations are visited. To account for the value of location, Figure~\ref{fig:vis-scores} reports the \emph{visibility score} of each product, defined as the sum of the probabilities of the locations at which the product is placed. 
The conclusions parallel those from Figure~\ref{fig:occ-scores}: high-revenue products are systematically assigned to the most visible locations, the allocation remains revenue-ordered, and ALG again produces higher duplication  than OPT. Together, these figures illustrate that both ALG and OPT favor duplicating top-revenue products and placing them in prime, high-visibility locations, but ALG does so more aggressively.

\begin{figure*}[]
  \centering
  \includegraphics[width=0.9\textwidth]{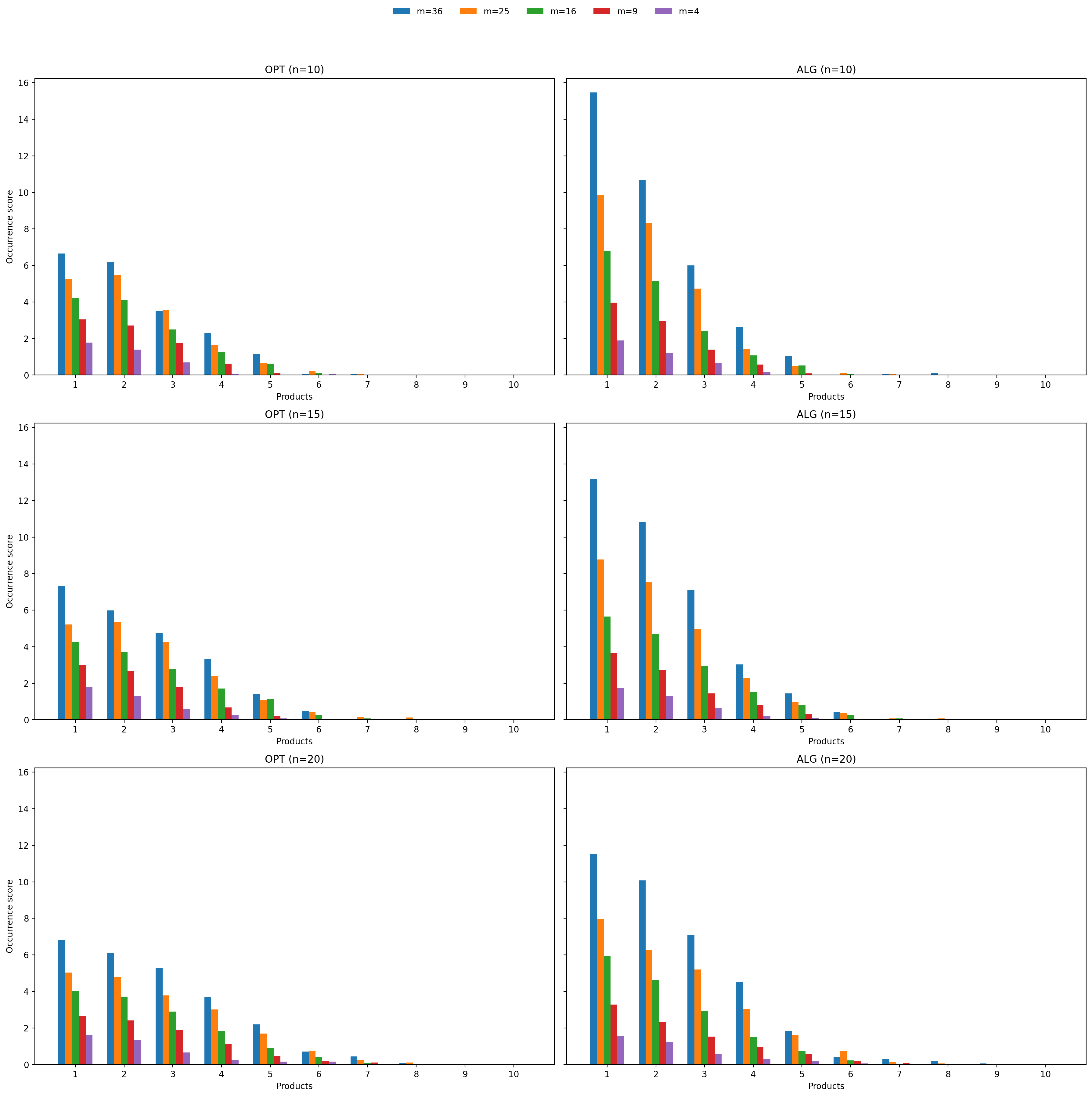}
  \caption{Occurrence score by product for OPT and ALG  for $n=10,15,20$.}
  \label{fig:occ-scores}
\end{figure*}

\begin{figure*}[]
  \centering
  \includegraphics[width=0.9\textwidth]{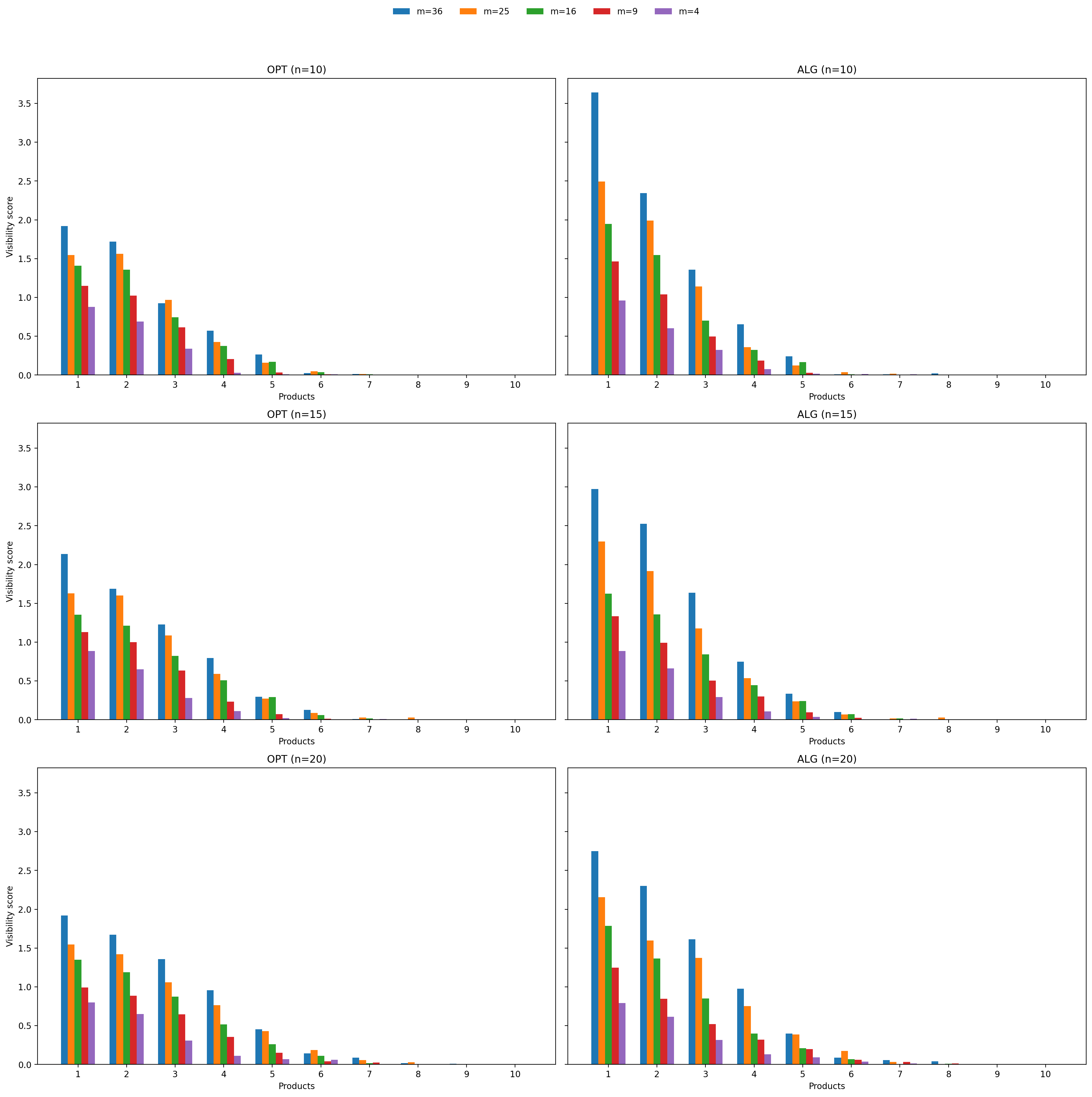}
  \caption{Visibility score by product for OPT and ALG for $n=10,15,20$.}
  \label{fig:vis-scores}
\end{figure*}

}

\end{APPENDICES}
\end{document}